\documentclass[pdflatex,sn-mathphys-num]{sn-jnl}

\usepackage{graphicx}%
\usepackage{multirow}%
\usepackage{amsmath,amssymb,amsfonts}%
\usepackage{mathrsfs}%
\usepackage[title]{appendix}%
\usepackage{xcolor}%
\usepackage{textcomp}%
\usepackage{manyfoot}%
\usepackage{booktabs}%
\usepackage{algorithm}%
\usepackage{algorithmicx}%
\usepackage{algpseudocode}%
\usepackage{listings}%
\usepackage{listings}
\usepackage{color}
\usepackage{xcolor}
\usepackage{graphicx}
\usepackage{float}
\usepackage{physics}
\usepackage{comment}
\usepackage{tcolorbox}
\usepackage{rotating} 
\usepackage{bbm} 
\usepackage{upgreek} 
\definecolor{dkblue}{rgb}{0,0.39,0}
\definecolor{gray}{rgb}{0.66,0.66,0.66}
\definecolor{mauve}{rgb}{0.91,0.33,0.50}
\definecolor{gold}{rgb}{1,0.84,0}

\usepackage{amsthm}%
\newtheorem{theorem}{Theorem}
\newtheorem{proposition}{Proposition}%
\newtheorem{lemma}{Lemma}%
\newtheorem{remark}{Remark}%
\newtheorem{definition}{Definition}%
\newtheorem{implication}{Implication}%

\begin{document}
\title{QI-MPC: A Hybrid Quantum-Inspired Model Predictive Control for Learning Optimal Policies}
\author*[1]{ \fnm{Muhammad Al-Zafar} \sur{Khan}}\email{Muhammad.Al-ZafarKhan@zu.ac.ae}

\author[1,2]{\fnm{Jamal} \sur{Al-Karaki}}\email{Jamal.Al-Karaki@zu.ac.ae}


\affil*[1]{\orgname{College of Interdisciplinary Studies}, Zayed University, \state{Abu Dhabi}, \country{UAE}}

\affil[2]{\orgname{College of Engineering}, The Hashemite University, \state{Zarqa}, \country{Jordan}}

\abstract{In this paper, we present Quantum-Inspired Model Predictive Control (QIMPC), an approach that uses Variational Quantum Circuits (VQCs) to learn
control polices in MPC problems. The viability of the approach is tested in five experiments: A target-tracking control strategy, energy-efficient building climate control, autonomous vehicular dynamics, the simple pendulum, and the compound pendulum. Three safety guarantees were established for the approach, and the experiments gave the motivation for two important theoretical results that, in essence, identify systems for which the approach works best.}

\maketitle

\section{Introduction}\label{introduction}

Model Predictive Control (MPC) stands as one of the most important advancements in modern control theory and combines optimization techniques with predictive modeling to anticipate system behaviors over a finite lookahead horizon. Model Predictive Control (MPC) has emerged as a cornerstone in modern control theory, offering a robust framework for managing complex, multivariable systems across diverse industries. By leveraging dynamic models to predict future system behaviors over a finite horizon, MPC enables the formulation of optimal control actions that adhere to system constraints. This predictive capability has been instrumental in applications ranging from chemical process optimization and energy management to advanced driver-assistance systems in the automotive sector. For instance, in building climate control, MPC has been utilized to reduce energy consumption while maintaining thermal comfort, demonstrating its efficacy in handling systems with inherent uncertainties and external disturbances [1].

Despite its advantages, the implementation of MPC can be computationally intensive, particularly for systems with nonlinear dynamics or when operating in real-time environments. To address these challenges, recent research has explored the integration of quantum computing principles into control strategies. Variational Quantum Circuits (VQCs), inspired by quantum mechanics, offer a promising avenue for enhancing computational efficiency in learning control policies. By combining the robustness of MPC with the computational advantages of VQCs, a Quantum-Inspired Model Predictive Control (QI-MPC) approach can be developed. This hybrid methodology aims to optimize control policies more effectively, particularly in complex systems such as autonomous vehicles, pendulum dynamics, and energy-efficient building management. The integration of quantum-inspired techniques into MPC not only addresses computational bottlenecks but also opens new pathways for theoretical advancements in control system design.

Mathematically, the objective of MPC is to solve the optimization problem at each timestep $t$
\begin{equation}
\underset{\mathbf{u}_{k}\in\mathcal{U}}{\min}\;J(\mathbf{x}_{t},\mathcal{U})=\underset{\mathbf{u}_{k}\in\mathcal{U}}{\min}\;\underbrace{\sum_{k=0}^{N-1}j(\mathbf{x}_{k},\mathbf{u}_{k})+V_{f}(\mathbf{x}_{N})}_{J(\mathbf{x}_{0},\mathcal{U})},\quad\text{for}\;k=0,1,2,\ldots,N-1
\end{equation}
where $\mathbf{u}_{k}$ are the control inputs/policies in the set of all feasible policies $\mathcal{U}\subseteq\mathbb{R}^{q}$, $J$ is the objective loss/cost function and $j$ is the stage cost, $V_{f}$ is the terminal cost, $\mathbf{x}_{t}$ are the state vectors in the set of state vectors $\mathcal{X}\subseteq{R}^{s}$ with $\mathbf{x}_{0}$ being the initial state, $\mathbf{x}_{N}\subseteq\mathcal{X}_{f}\subseteq\mathcal{X}$ is the terminal state, over a prediction horizon/lookahead $N$, subject to
\begin{equation}
\left\{
\begin{aligned}
&\mathbf{x}_{k+1}=f(\mathbf{x}_{k},\mathbf{u}_{k}),\quad k=0,1,2,\ldots N-1 \\
&\mathbf{x}_{k}\in\mathcal{X},\quad k=0,1,2,\ldots N-1 \\
&\mathbf{u}_{k}\in\mathcal{U},\quad k=0,1,2,\ldots N-1 \\
&\mathbf{x}_{N}\in\mathcal{X}_{f},\quad\text{if terminal constraints are used}
\end{aligned}
\right. 
\end{equation}
At each timestep $t$, the current state $\mathbf{x}_{t}$ is estimated and the open-loop control problem 
\begin{equation}
\mathbf{u}^{*}=\left(\mathbf{u}_{0}^{*},\mathbf{u}_{1}^{*},\ldots,\mathbf{u}_{N-1}^{*}\right)=\underset{\mathcal{U}}{\arg\min}\;J(\mathbf{x}_{t},\mathcal{U}),
\end{equation}
is solved. Thereafter, the first policy $\mathbf{u}_{t}=\mathbf{u}^{*}_{0}$ is applied and the horizon is sifted forward and the process is repeated. 

Unlike traditional control methods that react solely to current states, MPC leverages an internal dynamic model to simulate future system responses, repeatedly solving an optimization problem at each time step to determine optimal control actions. This approach allows MPC to handle complex constraints, nonlinearities, and multiple objectives simultaneously, making it particularly valuable in scenarios where system dynamics are well-characterized but challenging to control. The real-world impact of MPC extends across numerous industries, demonstrating remarkable versatility and effectiveness. In the petrochemical sector, it has revolutionized refinery operations by optimizing production while maintaining strict safety constraints \cite{richalet1993industrial,zanin2002integrating,qin2003survey,yuzgecc2010refinery,sildir2014economic,wei2017dynamic}. Automotive applications have seen MPC employed in advanced driver assistance systems \cite{wang2012driver,lefevre2015driver,cesari2017scenario,ercan2017modeling,sajadi2019nonlinear,musa2021review} and engine control units \cite{ortner2007predictive,stewart2008model,pizzonia2016robust,hrovat2012development}, significantly improving fuel efficiency and emissions compliance. In addition, has also proven its utility in power systems \cite{kouro2008model,jin2009model,martin2016corrective,liu2019model}, enabling more efficient integration of renewable energy sources by predicting and managing intermittent generation. Other fields of applications where MPC is an indispensable tool are aerospace flight control systems \cite{van2006combined,eren2017model,di2018real}, industrial robotics \cite{faulwasser2016implementation,carron2019data,gold2020model}, building climate control \cite{oldewurtel2012use,oldewurtel2013stochastic,wang2019data}, and medical drug delivery systems \cite{SOPASAKIS2014193,SHARIFI201713,7836354}, and various other applications. MPC has repeatedly delivered performance improvement, enhanced efficiency, and more robust operation in the face of uncertainties and system disturbances.

However, one major challenge is the generation of optimal control sequences/policies. This arises from the need to solve a constrained optimization problem repeatedly, often in real-time, under tight computational constraints. Specifically, the controller is required to find the control policies and must:
\begin{enumerate}
\item Minimize the cost function that balances control objectives.
\item Satisfy system dynamics constraints.
\item Respect state and input constraints.
\item Obtain the optimal control policies within the available sampling period.
\end{enumerate}
The difficulty becomes more pronounced when the system dynamics are nonlinear, when there are longer prediction horizons, there are higher-dimensional state and input/control spaces, the more complex the system constraints, and if there are multiple objective functions. Specifically, for discrete or mixed-integer decision variables, the problem becomes combinatorial in nature. The number of possible control sequences $u$ grows exponentially with the prediction horizon $N$ and the number of control inputs $m$, i.e. $\mathcal{O}(u^{mN})$. This creates a fundamental tension between computational traceability, control performance, and prediction accuracy. Further, industrial applications require solutions within strict time constraints. For fast dynamic systems such as Robotics and Automotive applications, the optimization must be completed within milliseconds, creating severe demand on the controller. In addition, the nature of the objective may create issues of nonconvexity. Specifically, multiple local minima may exist, global optimality guarantees may be difficult to provide or outright impossible, and solution methods may fail to converge. Existing methods like explicit MPC, fast numerical solvers (such as interior-point and active-set methods), approximation methods, hierarchical MPC, and distributed MPC, although reliable, suffer from problems related to the state-space dimensioning growing exponentially as the problem dimensions grow, limited computational memory, scaling limitations, coordination challenges leading to suboptimality, among others. 

This begs the question: ``Is there an efficient way in which we can determine the optimal control sequences that balances computational efficiency and control performance?'' The objective of this research is to answer this exact question. Of course, we are not the first to attempt to address this problem. In Tab. \ref{tab1}, we discuss some of the fortitudes and limitations of the current methods.

\begin{table}[htp!]
\centering
\begin{tabular}{p{5cm}p{4cm}p{5cm}}
\hline 
\textbf{Method} &\textbf{Strengths} &\textbf{Weaknesses} \\
\hline
Multi-parametric programming \cite{KOURAMAS20111638} &Converts online optimization to fast lookup table evaluation. &State space partitioning grows exponentially with state dimension. \\
 & &Has a practical limit to the number of states that it can handle. \\
\hline 
Move blocking \cite{SHEKHAR201527} &Reduces decision variables by fixing control inputs over multiple time steps. &Introduces suboptimality that Is difficult to quantify. \\
 & &If the blocking scheme is poorly chosen, it can lead to violations of the constraints. \\
 & &Reduces the flexibility of the control. \\
\hline 
Input parametrization \cite{7099558} & Represents control trajectory with basis functions to reduce parameters.
 &Basis function selection is not general, i.e. it depends on the specifics of the problem being studied. \\
 & &Not guaranteed to capture the optimal control profiles. \\
 & &Struggles in handling nonlinear constraints. \\
\hline
Machine learning-based MPC \cite{Wu2019} &Uses offline learning to speed up online. optimization. &Struggles with out-of-distribution generalization. \\
 & &Training data does not cover all operational scenarios. \\
 & &Difficult to establish safety guarantees. \\ 
\hline 
Neural network/deep learning-based MPC \cite{10419329} &Fast execution after training the model. 
 &The current black-box nature of deep learning limits the interpretability and the verifiability of the results. \\
 & &Model training requires extensive data and simulations. \\
 & &Satisfying all, or even some, of the constraints, is not guaranteed. \\
\hline
Stochastic MPC \cite{7740982} &Accounts for probabilistic uncertainties in the system model. &Significantly increases computational complexity due to the randomness element. \\
 & &Requires an accurate characterization of the uncertainty. \\
 & &Relies on scenario-based approaches, which scales very poorly. \\
\hline 
Single/multiple shooting-based methods \cite{KIRCHES2012540} &Creates structure in the optimization problem, which is efficient for obtaining solutions. &Struggles with nonlinear systems. \\
 & &Multiple shooting increases the problem size. \\
\hline
Continuation methods \cite{BAAYEN201973} &Uses previous solution to warm-start current problem. &May not recover from poor initial guesses. \\
 & &Performance degrades and becomes suboptimal when system conditions change rapidly. \\
\hline 
\end{tabular}
\caption{Strengths and weaknesses of the current approaches to addressing the problem of selecting optimal control policies in MPC.}
\label{tab1}
\end{table}

Quantum Computing (QC), which has proved its utility in many Machine Learning (ML) applications \cite{innan2023enhancing,innan2023classical,innan2024quantum1,khan2024brain,khan2024predicting,innan2024financial,innan2024quantum2}, potentially offers an alternative approach to solving complex optimization problems in MPC through the use of Variational Quantum Circuits (VQCs) -- A sequence of quantum gates with tunable parameters. By encoding control sequence selection as a quantum optimization problem, VQCs can admissably explore vast solution spaces simultaneously through quantum superposition and entanglement, providing a significant computational advantage over classical methods for complex, nonlinear, and high-dimensional control problems. This quantum approach is particularly promising for systems with combinatorial explosion challenges, where the quantum circuit can naturally represent the complex interactions between control decisions across the prediction horizon. Initial experimental results suggest that quantum algorithms could eventually achieve quadratic or even exponential speedups for certain control optimization problems, explored in greater depth in Sec. \ref{outperformance of classical mpc}, enabling real-time MPC implementation for previously intractable systems. As quantum hardware continues to mature beyond the current noisy intermediate-scale quantum (NISQ) era, this quantum advantage could create a competitive advantage of control strategy selection in fields like autonomous vehicle navigation, chemical process control, and complex robotic systems where traditional methods hit computational barriers. The benefit of this approach are threefold:
\begin{enumerate}
\item \textbf{Facilitation of Parallelism:} Allows for multiple control policies are evaluated simultaneously.
\item \textbf{Enhanced Search:} Entanglement allows better search in the control space.
\item \textbf{Energy Efficiency:} Potentially, computational costs are reduced when solving high-dimensional search problems, and this, in turn, reduces the energy required for computation.  
\end{enumerate}
However, it is worth mentioning that there are some studies in which QC is not applied to solve the MPC problem, but methods inspired by QM are applied. A noteworthy mention is in \cite{de2023nonlinear} whereby the authors attempted to solve the circular restricted three-body problem in real-time on a spacecraft using diffusion Monte Carlo methods that were inspired by quantum mechanics in order to solve the stochastic optimal control problem. It was shown that the method successfully generated feasible trajectories, despite initial condition uncertainties and disturbance accelerations. 

This paper is divided as follows:

In Sec. \ref{related work}, we explore studies that have attempted a similar feat to the one we want in this research.

In Sec. \ref{qcmpc}, we discuss our proposed approach and establish some safety guarantees.

In Sec. \ref{experiments}, we apply our approach to five application areas to give a numerical (in)validation of our approach.

In Sec. \ref{physical results}, we establish some theoretical results based on our observations from the experiments carried out.

In Sec. \ref{conclusion}, we provide closing arguments whereupon we reflect on the results obtained, discuss limitations, and ideate areas of exploration that stem from this study.


\section{Related Work}\label{related work}

Hybrid approaches to MPC have not been explored, to our knowledge, in great detail. We review some of the prominent approaches, with the best in our opinion being the work in conducted in \cite{inoue2020model} whereby the authors argue that due to the combinatorial optimization formulation of MPC problems as being NP-hard, a reformulation of the problem as a QUBO problem can lead to computational efficiency and better results. In particular, the authors use the D-Wave quantum annealer to solve two MPC problems with finite input values: The classical spring-mass-damper system and dynamic audio quantization. The results obtained from this approach were shown to be superior to the standard MPC approach. 

The work in \cite{goldschmidt2022model} shows how MPC enhances quantum control by leveraging real-time feedback, handling constraints, and improving robustness in systems where the Hamiltonian is not well-characterized. Specifically, the authors apply MPC strategies to control quantum processors containing superconducting qubits and cold atoms. It was shown that the approach works well, achieving high fidelity, even with imperfect models, and it can be combined with other contemporary control techniques.    

In \cite{clouatre2022model}, the authors introduce a novel approach whereby they combine Bayesian Hamiltonian learning with MPC for quantum systems. Specifically, the approach studies systems whereby the dynamics are partially known or completely unknown, and the authors design open-loop quantum control sequences that maximize fidelity to a target quantum gate, even with uncertain Hamiltonian parameters. Sequential Monte Carlo is used to estimate the posterior probability distribution and using samples generated from the posterior, stochastic MPC is applied. The proposed approach was tested on a single qubit system with a Hamiltonian that was the linear combination of Pauli gates, and the goal was to find the optimal values of the unknown parameters that best achieve the targeted Hadamard gate.  

In \cite{novara2024quantum}, the authors reformulate nonlinear MPC problems in a polynomial form and use quantum annealers, among other quantum approaches, to solve the underlying problem. The authors claim that their approach has the ability to revolutionize how such nonlinear problems are solved by reducing the computational time. However, no numerical results are provided in order to validate or invalidate this claim. 


\section{Quantum-Informed Model Predictive Control}\label{qcmpc}
In this section, we present and discuss the thought process behind our proposed approach and thereafter establish some safety guarantees of the algorithm.

\subsection{The QI-MPC Algorithm}
We introduce Quantum-Inspired Model Predictive Control (QI-MPC), an approach that operates in a hybrid quantum-classical loop, where a VQC generates control inputs while classical computation handles system evolution and optimization. 

Our thought process in presenting this method starts with encoding the current classical system state into a quantum state by some encoding scheme. Ubiquitous encoding schemes include amplitude encoding, phase encoding, and quantum feature maps, among others. The resulting quantum state is then processed by a parameterized unitary operator -- the ansatz -- which applies a series of rotational and entangling gates to transform the state. The control outputs are then measured in the computational basis, yielding a set of scalar values that represent the proposed policies. These raw outputs are clipped to enforce physical constraints, ensuring they lie within the admissible range. Moreover, by clipping the controls, exploding gradients are avoided in future steps where the gradients are calculated, and this minimizes instabilities and divergences. 

Once the constrained controls are obtained, they are applied to the system dynamics to compute the next state. Thereafter, in order to balance the deviation of the new state from the target state, the regularized loss function is calculated. Regularization is added to the loss function to penalize any large deviations from the difference between successive control sequences. It can easily be extended to include losses from the learned parameters in the VQC itself. We advocate for the use of gradient-based methods, like the parameter-shift rule, for updates because they are quantum-compatible, enable analytical gradient computation without backpropagation, and avoid complex automatic differentiation. However, this can easily be adjusted to any suitable parameter update rule.

In order to make our approach more robust, we adopt the receding-horizon approach, as in classical MPC, where only the first control action is executed, and the optimization window shifts forward in time. This allows the process to continuously adapt to new state information, maintaining rigidity against uncertainties and disturbances that may arise.

We formally summarize these thought processes in Algorithm \ref{algo1}.

\begin{algorithm}[H]
\caption{QI-MPC}
\label{algo1}
\small{
\begin{algorithmic}[1]
\State \textbf{input} system dynamics $f(\mathbf{x}_{k},\mathbf{u}_{k})$, initial state $\mathbf{x}_{0}$, target state $\mathbf{x}_{\text{target}}$, prediction horizon $N$, control timesteps $T\gg N$, number of qubits $n$, optimization bounds $\left[u_{\min},u_{\max}\right]$, control dimensions $m$, penalization parameter $0\leq\lambda\leq1$
\State \textbf{output} optimal control sequence $\mathbf{u}^{*}$, optimal state trajectory $\mathbf{x}^{*}$
\State \textbf{initialize} quantum device with $n$ qubits, VQC $U(\boldsymbol{\theta})$, classical optimizer, tolerance $\epsilon$ (optional)
\Repeat
\For {$k$ from $0$ to $T$} \Comment{iterating over all control timesteps $T$}
\State encode current state into quantum state $\ket{\psi(\mathbf{x}_{k})}=\text{encode}(\mathbf{x}_{k})\ket{0}^{\otimes n}$
\State apply the ansatz/unitary operator $\ket{\psi(\mathbf{x}_{k},\boldsymbol{\theta})}=U(\boldsymbol{\theta})\ket{\psi(\mathbf{x}_{k})}$  
\State perform measurement on the control outputs 
\begin{equation*}
\mathbf{u}_{k}=\left(\bra{\psi(\mathbf{x}_{k},\boldsymbol{\theta})}Z_{1}\ket{\psi(\mathbf{x}_{k},\boldsymbol{\theta})},\bra{\psi(\mathbf{x}_{k},\boldsymbol{\theta})}Z_{2}\ket{\psi(\mathbf{x}_{k},\boldsymbol{\theta})},\ldots,\bra{\psi(\mathbf{x}_{k},\boldsymbol{\theta})}Z_{m}\ket{\psi(\mathbf{x}_{k},\boldsymbol{\theta})}\right)
\end{equation*}
\State Limit the controls to satisfy the constraints $\tilde{\mathbf{u}}_{k}=\text{clip}(\mathbf{u}_{k},u_{\min},u_{\max})$
\State apply the controls to determine the system evolution $\mathbf{x}_{k+1}=f(\mathbf{x}_{k},\tilde{\mathbf{u}}_{k})$
\State compute the loss $\mathcal{J}=||\mathbf{x}_{k+1}-\mathbf{x}_{\text{target}}||^{2}+\lambda||\tilde{\mathbf{u}}_{k}||^{2}$
\State update the trainable parameters
\begin{equation*}
\boldsymbol{\theta}\gets\text{classical\_optimizer}(\nabla\boldsymbol{\theta})    
\end{equation*}
\State shift prediction horizon to perform the receding horizon update
\EndFor
\Until {$N$ is reached or $||\mathbf{x}-\mathbf{x}_{\text{target}}||<\epsilon$}
\State \textbf{return} optimal policies and trajectories
\begin{equation*}
\mathbf{u}^{*}=\left(u_{0}^{*},u_{1}^{*},\ldots,u_{T-1}^{*}\right),\quad \mathbf{x}^{*}=\left(x_{0}^{*},x_{1}^{*},\ldots, x_{T-1}^{*}\right)
\end{equation*}
\end{algorithmic}
}
\end{algorithm}

Algorithm \ref{algo1} repeats until either the prediction horizon is exhausted or the state converges to the target with some tolerable range $\epsilon$.  

We believe that the strength on our QI-MPC approach lies in its ability to handle high-dimensional state spaces efficiently, leveraging quantum parallelism to explore multiple control policies simultaneously. However, the performance of Algorithm \ref{algo1} is contingent on several factors, including the smoothness of the system dynamics and the alignment of quantum measurement timescales with the control requirements. As an illustrative example, if systems with rapid oscillations or discontinuous dynamics are explored, they may challenge the algorithm's stability, as the quantum-classical feedback loop might not react swiftly enough. Similarly, the energy scales of macroscopic systems could render the quantum-generated controls ineffective due to their inherently small magnitudes.

\subsection{Safety Guarantees of QI-MPC}\label{safety guarantees}
To ensure the QI-MPC algorithm operates safely and reliably, we establish formal safety guarantees through constraint satisfaction, stability analysis, and quantum measurement robustness.

\begin{enumerate}
\item \textbf{Formal Verification I:} The control inputs $\tilde{\mathbf{u}}_{k}$ always remain within bounds. This is because $\tilde{\mathbf{u}}_{k}$ is projected to $\left[u_{\min},u_{\max}\right]\;\forall k$, i.e.,
\begin{equation}
\tilde{\mathbf{u}}_{k}=\text{clip}(\mathbf{u}_{k},u_{\min},u_{\max})=\max\left\{u_{\min},\min\left(\mathbf{u},u_{\max}\right)\right\}. 
\end{equation}
The physical interpretation of this guarantee is that it prevents actuator saturation and ensures that the policies are realizable.
\item \textbf{Formal Verification II:} For closed-loop systems, QI-MPC is stable. Assuming that the system dynamics $f(\mathbf{x}_{k},\mathbf{u}_{k})$ is Lipschitz continuous, with Lipschitz constant $K$, i.e.
\begin{equation*}
||f(\mathbf{x},\mathbf{u})-f(\mathbf{x}',\mathbf{u}')||\leq K\left(||\mathbf{x}-\mathbf{x}'||+||\mathbf{u}-\mathbf{u}'||\right),
\end{equation*}
and the cost $j$ at each stage is positive definite, then by the stability theorem, the algorithm is stable,
\begin{equation}
\lim_{k\to\infty}\mathbf{x}_{k}=\mathbf{x}_{\text{target}}.
\end{equation}
This is because the cost at each stage $j$ is a control Lyapunov function and the receding horizon ensures recursive feasibility, i.e. if the system is stable at $k=0$, then it remains stable $\forall k=1,2.\ldots N-1$.
\item \textbf{Formal Verification III:} The output control sequence is reliable despite the noise arising from the quantum circuit. Each time a measurement is done in a quantum circuit, measurement noise $\eta\sim\mathcal{N}(0,\sigma^{2})$ (which we assume is normally distributed, for all intents and purposes, with mean $0$ and variance $\sigma^{2}$), i.e., 
\begin{equation*}
\mathbf{u}_{k}=\bra{\psi(\mathbf{x}_{k},\boldsymbol{\theta})}Z_{i}\ket{\psi(\mathbf{x}_{k},\boldsymbol{\theta})}+\eta_{k}.
\end{equation*}
By Hoeffding's inequality, the probability of the deviation between the control and its average value is bounded by
\begin{equation*}
\mathbb{P}\left(|\mathbf{u}_{k}-\mathbb{E}(\mathbf{u}_{k})|\geq\epsilon\right)\leq 2e^{-2\epsilon^{2}\mathcal{M}},    
\end{equation*}
for $\mathcal{M}$ readouts (executions of the quantum circuit and measurement of the final state). For sufficiently large $\mathcal{M}$, the error is exponentially suppressed and we have that
\begin{equation}
 ||\mathbf{u}_{k}-\mathbb{E}(\mathbf{u}_{k})||\leq\epsilon.    
\end{equation}
\end{enumerate}


\section{Experiments}\label{experiments}
In this section, we present the application of Algorithm \ref{algo1} to five experiments: The first with dynamics described by a target-tracking proportional control strategy, the second to energy efficient climate control strategies in smart buildings, the third to autonomous vehicles, the fourth to a simple pendulum control problem, and the fifth to a compound pendulum control problem. For each of the experiments, we use the quantum ansatz to find the optimal control sequence and thereafter graphically investigate the performance of the method by examining the associated quantum circuit, the evolution of the states, the evolution of the controls, and the loss curve. 

\subsection{Experiment 1: A Target-Tracking Proportional Control Strategy}\label{experiment1}
The initial states $\mathbf{x}_{0}\in\mathbb{R}^{3}$ are encoded into quantum states using rotation gates according to
\begin{equation}
\mathbf{x}_{0}\to\ket{\mathbf{\psi}(\mathbf{x})}=\prod_{i=1}^{3}RY(\pi x_{i})\cdot RX\left(\pi\left(x_{i}+\frac{1}{2}\right)\right)\cdot RZ\left(\pi\frac{x_{i}}{2}\right)\ket{0}^{\otimes n},
\end{equation}
where $RX, RY, RZ$ are arbitrary rotations about the $x,y,z$-axes, respectively, defined as
\begin{equation}
RX(\theta)=
\begin{pmatrix}
\cos\theta/2 &-\imath\sin\theta/2 \\
-\imath\sin\theta/2 &\cos\theta/2
\end{pmatrix}, \quad 
RY(\theta)=
\begin{pmatrix}
\cos\theta/2 &-\sin\theta/2 \\
\sin\theta/2 &\cos\theta/2
\end{pmatrix}, \quad 
RZ(\theta)=
\begin{pmatrix}
e^{-\imath\theta/2} &0 \\
0 &e^{\imath\theta/2}
\end{pmatrix},
\end{equation}
for an arbitrary angle $\theta$. Thereafter, the ansatz, comprised of parametrized quantum layers, is applied where trainable rotation and entanglement gates are used, according to
\begin{equation}
U(\boldsymbol{\theta})=\prod_{t=1}^{T}\left(\prod_{q=1}^{n}\text{Rot}(\theta_{t,q})\cdot\prod_{\langle i,j\rangle}\text{CNOT}_{i,j}\right),
\end{equation}
where $\theta_{t,q}$ are the trainable parameters, $\text{Rot}(\theta)=RZ(\theta_{1})\cdot RY(\theta_{2})\cdot RZ(\theta_{3})$, and the CNOT gates are used to create entanglement between qubit pairs $\langle i,j\rangle$. Lastly, the control inputs $\mathbf{u}\in\mathbb{R}^{3}$ is obtained by taking expectation values of the Pauli $Z$ operators according to
\begin{equation}
u_{k}=\bra{\psi(\mathbf{x})}Z_{k}\ket{\psi(\mathbf{x})},\quad\text{for}\;k=1,2,3.
\end{equation}
Within this experiment, we assume the simplified dynamics given by the target-tracking/reference-following control strategy
\begin{equation}
\mathbf{x}(t+1)=\mathbf{x}(t)+\alpha\left[\mathbf{u}(t)-\mathbf{x}(t)\right],
\end{equation}
with a learning rate $\alpha=0.1$, and we minimize the MSE loss
\begin{equation}
\mathcal{L}(\mathbf{x};\mathbf{x}_{\text{target}})=||\mathbf{x}-\mathbf{x}_{\text{target}}||_{2}^{2}. 
\end{equation}

\begin{figure}[H]
\centering
\includegraphics[width=1.2\linewidth]{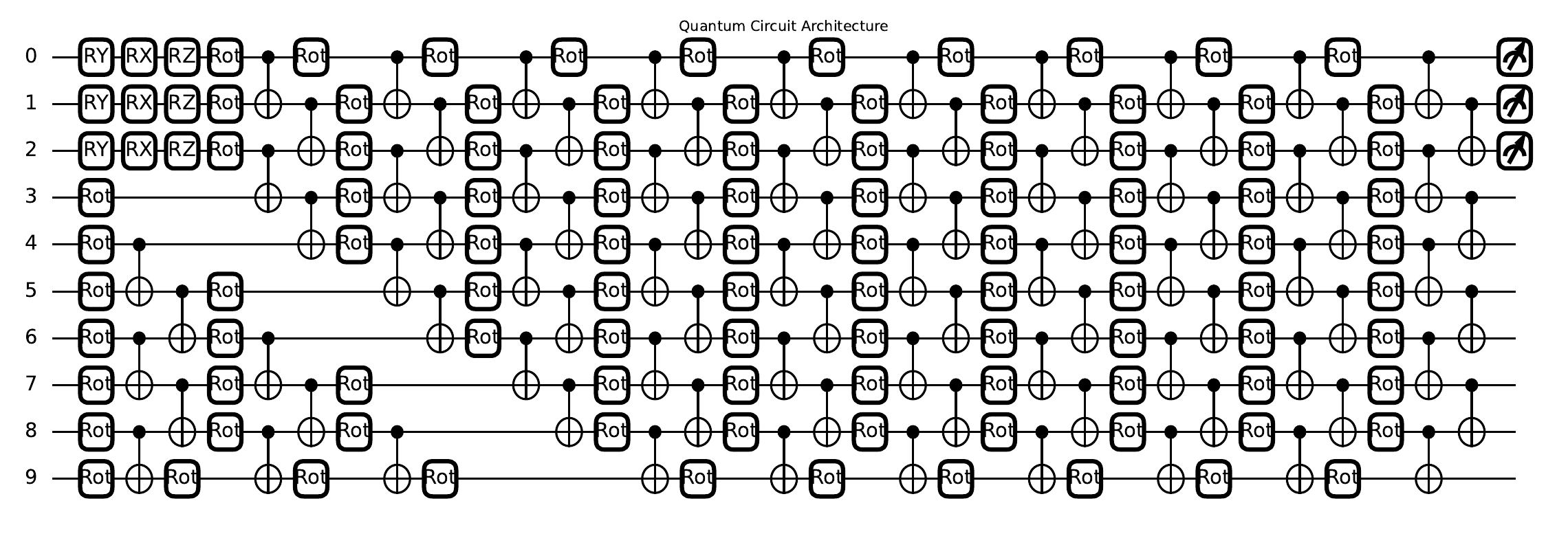}
\caption{The quantum circuit used to determine the optimal control actions for the state dynamics $\mathbf{x}_{k}(t+1)=\mathbf{x}_{k}(t)+\alpha\left[\mathbf{u}_{k}(t)-\mathbf{x}_{k}(t)\right]$.}
\label{fig4}
\end{figure}

\begin{figure}[H]
\centering
\includegraphics[width=1\linewidth]{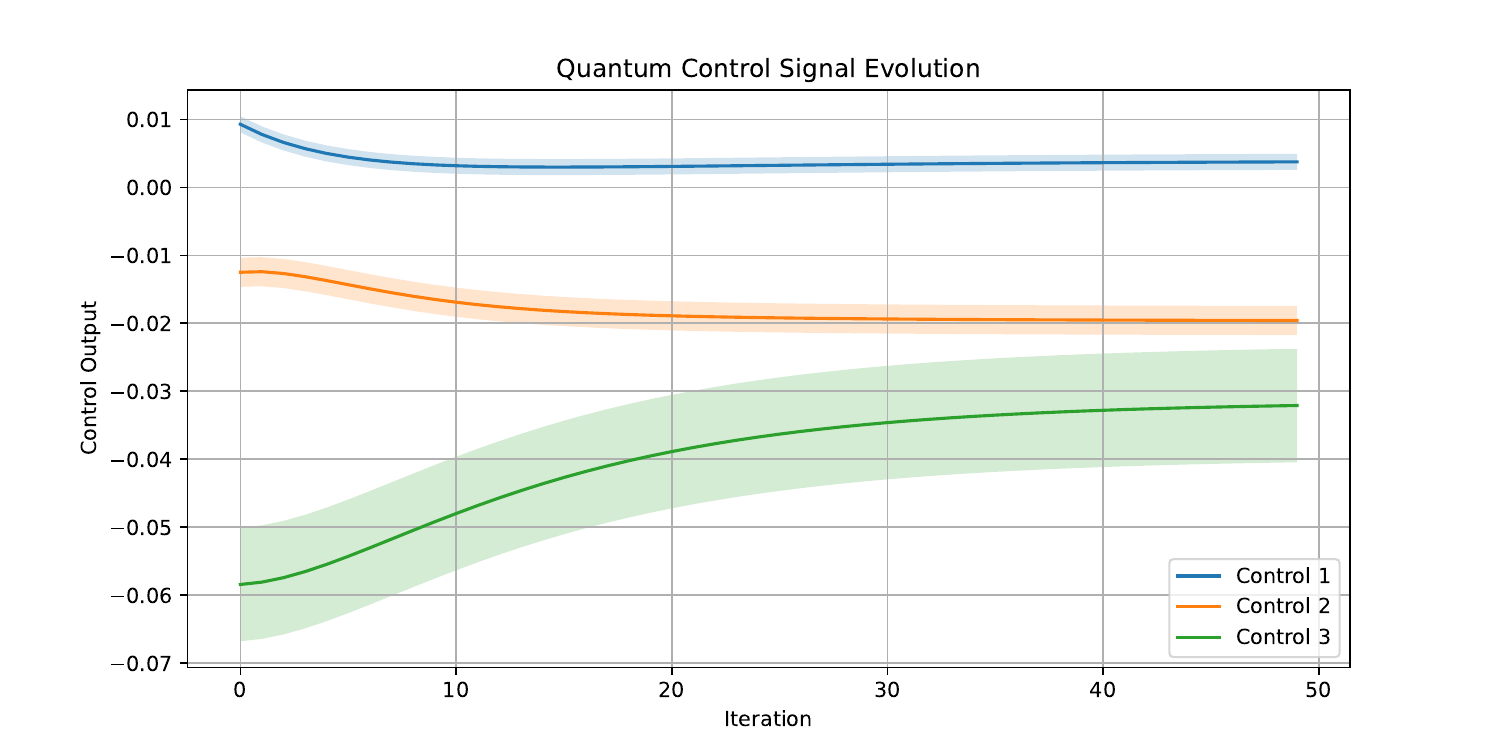}
\caption{Control signals for the controls $\mathbf{u}=\left(u_{1},u_{2},u_{3}\right)\in\mathbb{R}^{3}$.}
\label{fig1}
\end{figure}

\begin{figure}[H]
\centering
\includegraphics[width=1\linewidth]{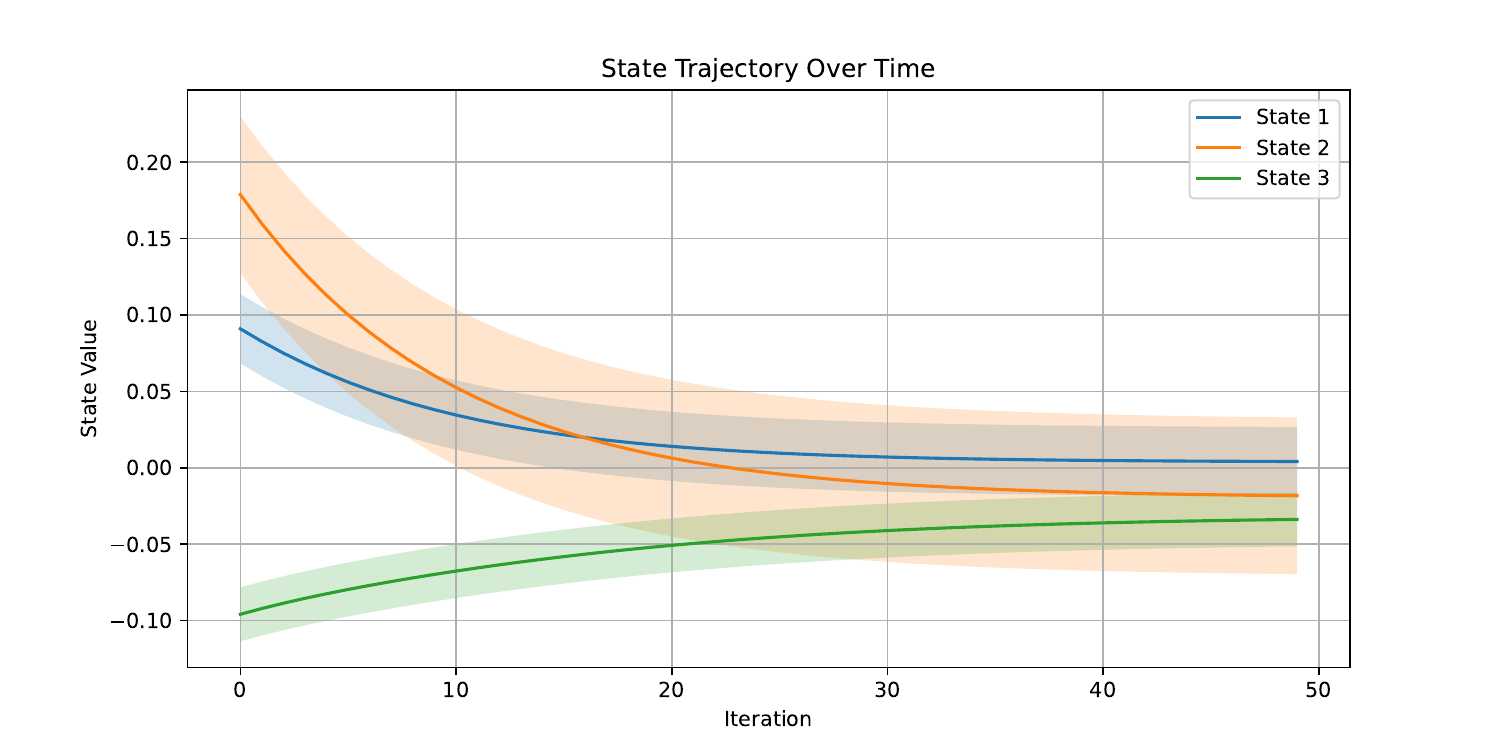}
\caption{State trajectories $\mathbf{x}=\left(x_{1},x_{2},x_{3}\right)\in\mathbb{R}^{3}$.}
\label{fig3}
\end{figure}

\begin{figure}[H]
\centering
\includegraphics[width=1\linewidth]{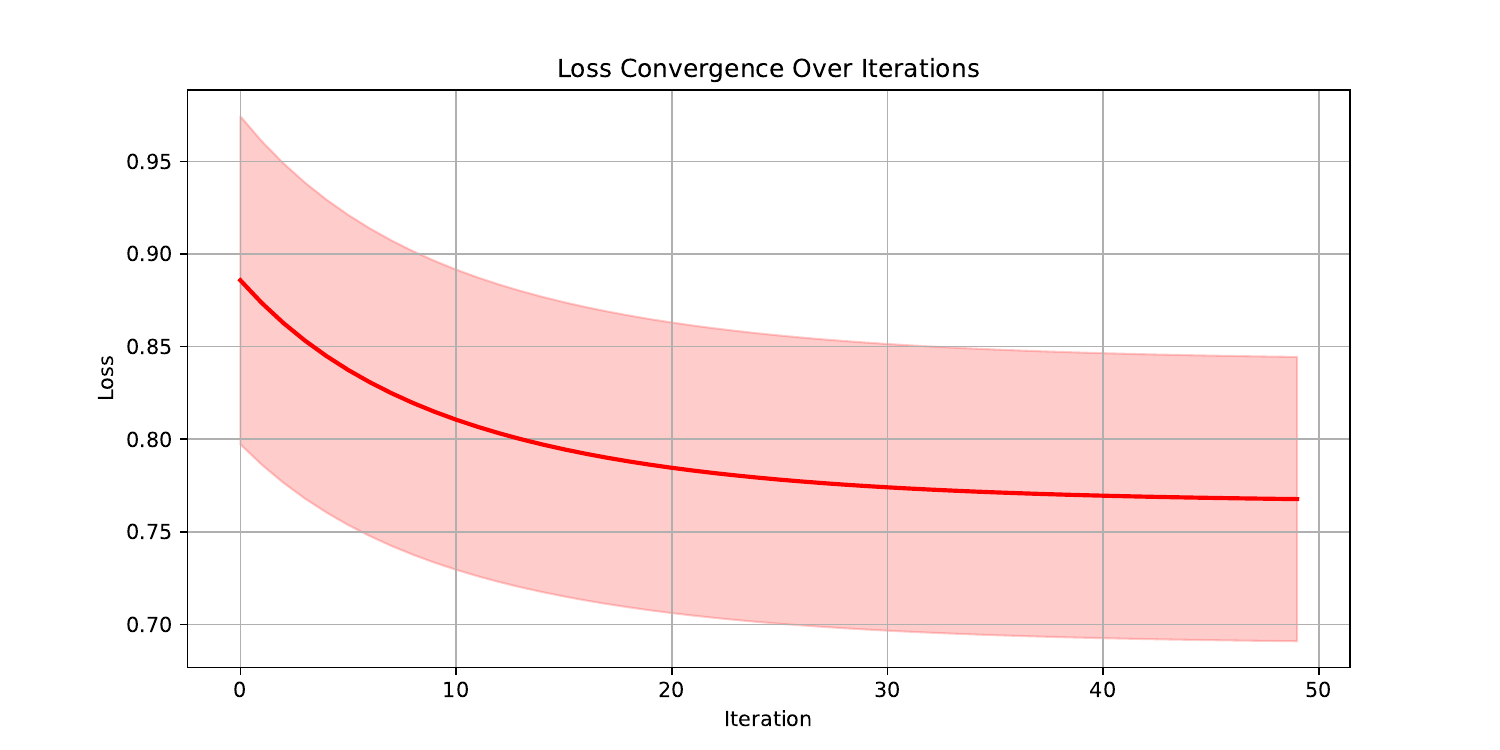}
\caption{Loss function convergence.}
\label{fig2}
\end{figure}

The parametrized quantum circuit in Fig. \ref{fig4} is used to compute quantum control signals for guiding the dynamical system toward a target state. Specifically, the circuit is composed of 10 qubits (labeled 0 to 9). The circuit is composed of single-qubit ($RX, RY, RZ$) and two-qubit (CNOT) gates in order to create entanglement and measurement operators on wires 0, 1, and 2. The single qubit rotation gates act as state encoders in order to map the classical states to quantum states. The parametrized layers consist of arbitrary rotation gate operators (Rot) defined by
\begin{equation}
\text{Rot}(\Phi_{1},\Phi_{2},\Phi_{3})=RZ(\Phi_{1})RY(\Phi_{2})RX(\Phi_{3}),
\end{equation}
where $\Phi_{1},\Phi_{2},$ and $\Phi_{3}$ act as trainable parameters. The circuit follows a repetitive pattern by applying rotation gates to all qubits and CNOT gates in order to create entanglement by flipping the target qubit conditioned on the control qubit's state. At the end of the circuit, three qubits, each corresponding to the respective control output, are measured using Pauli-$Z$ expectation values. The measured values are mapped back to classical control values, which are then used in the MPC loop to update the system's state. 

In Fig. \ref{fig1}, we observe the evolution of quantum control signals over multiple iterations in our proposed scheme. The three curves correspond to three control variables $\mathbf{u}=\left(u_{1},u_{2},u_{3}\right)$ over 50 iterations. The shading around each curve likely represents the confidence interval at one standard deviation of the control values; this provides insight into the stability and uncertainty of the learned control signals. The curves exhibit smooth convergence with no discontinuities; this indicates that the quantum controller successfully adjusts the control signals over time. $u_{1}$ (blue) and $u_{2}$ (orange) appear to stabilize rapidly within the first 10 to 15 iterations. This suggests they are closer to an optimal control solution early in the training process. However, $u_{3}$ (green) gradually increases over time. This indicates that this control variable undergoes a more dynamic adjustment. Further, the variance for $u_{1}$ and $u_{2}$ is low, as is evidenced by the feint shading around these lines; this indicates that they are more stable. Conversely, the shading around $u_{3}$ is more pronounced. This indicates that the model struggles more with this control variable before stabilizing. We note that these control signals are derived from quantum circuit measurements in Fig. \ref{fig4}, meaning they are not directly optimized in a classical sense but instead obtained from quantum state evolution. Thus, the controller adapts iteratively based on feedback, and the structure of the quantum circuit and measurement noise influence the evolution of these signals.

From Fig. \ref{fig3}, we see that the states converge toward a steady state over time. This means that the controller is effectively stabilizing the system. $x_{1}$ (blue) and $x_{2}$ (orange) exhibit an exponential decay pattern, which is indicative of a stable feedback control system, and we can say that these two curves have profiles that are approximated by
\begin{equation}
x_{k}\approx x_{\infty}+\left(x_{0}-x_{\infty}\right)e^{-\beta k},\quad\text{for}\;k=1,2,
\end{equation}
where $x_{\infty}$ is the long-term steady state value achieved and $\beta$ is the decay constant. As with Fig. \ref{fig1}, the shaded regions around the curves indicate confidence intervals, signifying uncertainty bounds, in the predicted state values. We see that $x_{2}$ has the highest initial variance, indicating it is initially more sensitive to control fluctuations, while $x_{1}$ and $x_{3}$ have smaller uncertainty bounds, implying that their trajectories are more predictable. Overall, these suggest that the quantum controller learns an effective control strategy over time and reduces the system variability, and the control signals are effective in steering the system states toward a desired target.

From Fig. \ref{fig2}, we see that the loss function decreases over iterations exhibiting almost a logarithmic-like or exponential decay-like profile, confirming that the quantum-based controller is learning an effective control policy and that the optimization process follows a gradient-based improvement strategy. Further, the profile suggests that the loss asymptotically approaches a lower bound, indicating that further improvements become marginal after a certain number of iterations. Around iteration 40, the decrease in loss is minimal, implying that further iterations offer diminishing improvements. The final loss level of $J\approx 0.75$ suggests that there may still be room for improvement if a better control policy is found. We attribute the quantum controller not reaching absolute zero error due to the stochasticity and noise due to quantum decoherence in the circuit. 

\subsection{Experiment 2: Energy Efficient Building Climate Control}\label{experiment4}
We consider the scenario of modeling the temperature of a room using an energy-efficient climate control in a smart building. The goal is to optimize the HVAC (Heating, Ventilation, and Air Conditioning) system to minimize energy consumption while maintaining thermal comfort. The temperature dynamics of the room can be described using the heat balance equation derived from thermal resistance-capacitance (RC) models
\begin{equation}\label{heating eqn}
C\frac{dT_{\text{room}}}{dt}=\frac{T_{\text{outdoor}}-T_{\text{room}}}{R}+P_{\text{HVAC}}+Q_{\text{solar}}+Q_{\text{occupants}},
\end{equation}
where $C$ is the thermal capacitance of the building, $T_{\text{room}}$ is the room temperature, $T_{\text{outdoor}}$ is the outdoor temperature, $R$ is the thermal resistance between the inside and outside air, $P_{\text{HVAC}}$ is the heating (or cooling) power of the HVAC system, $Q_{\text{solar}}$ is the heat gain from solar radiation, and $Q_{\text{occupants}}$ is the heat gain from the people inside the building. Using Euler's method, we discretize \eqref{heating eqn} to 
\begin{equation}
x_{k+1}=x_{k}+\left(\frac{T_{\text{out},k}-x_{k}}{RC}+\frac{u_{k}}{C}+\frac{Q_{\text{solar},k}+Q_{\text{occupants},k}}{C}\right)\Delta t,
\end{equation}
where we have defined the state variable $x_{k}\overset{.}{=}T_{\text{room},k}$ and the control $u_{k}\overset{.}{=}P_{\text{HVAC},k}$. The QMPC controller is used to minimize energy consumption while ensuring the temperature stays within comfort limits. Therefore, we write the cost function as
\begin{equation}\label{loss for room heating}
J=\sum_{k=0}^{N}\left(||x_{k}-x_{\text{set point}}||^{2}+\lambda_{1}||u_{k}||^{2}+\lambda_{2}||\Delta u_{k}||^{2}\right),
\end{equation}
where $x_{\text{set point}}$ is the desired temperature/temperature at which the thermostat has been set in the room, $\Delta u_{k}=u_{k}-u_{k-1}$ is the change in HVAC power for heating/cooling, $J_{\text{quantum}}$ is the loss from the VQC, and $0\leq\lambda_{1},\lambda_{2}\leq 1$ are the penalization weights. Specifically, $\lambda_{1}$ controls the penalty term $||u_{k}||^{2}$ by discouraging large control actions and $\lambda_{2}$ is the weight associated with the smoothness penalty $||\Delta u_{k}||^{2}$ by ensuring smooth control transitions over time. The loss function in \eqref{loss for room heating} is solved subject to the constraints:
\begin{enumerate}
\item \textbf{HVAC Actuation Limits:} The HVAC system has actuators with an operating range that serve to prevent damage to the system, ensure proper operation, and maintain system efficiency. Mathematically, we express this as $0\leq P_{\text{HVAC}}\leq P_{\max}$.
\item \textbf{Temperature Comfort Range:} The room temperature must stay within a comfortable range. Typically, smart systems try to maintain a range of $20-24^{\circ}\text{C}$. Mathematically, we express this as $T_{\min}\leq x_{k}\leq T_{\max}$. 
\end{enumerate}

We simulate the system above using the parameter values $R=0.5\;\text{K/W}, C=1\;\text{J/K}, P_{\max}=10\;\text{W}, T_{\min}=20^{\circ}\text{C}, T_{\max}=24^{\circ}\text{C}, T_{\text{outdoor}}=15^{\circ}\text{C}, Q_{\text{solar}}=5\;\text{W}, Q_{\text{occupants}}=3^{\circ}\text{C}$, we chose a temperature decay value in the dynamics equation to be 0.1. The model was trained for 200 iterations, with penalty weights $\lambda_{1}=0.01$ and $\lambda_{2}=0.005$. Lastly, we take the target states are $x_{k}=22$ for $k=1,2,3$. This implies that the set point temperature is $22^{\circ}\text{C}$. 

\begin{figure}[H]
\centering
\includegraphics[width=1\linewidth]{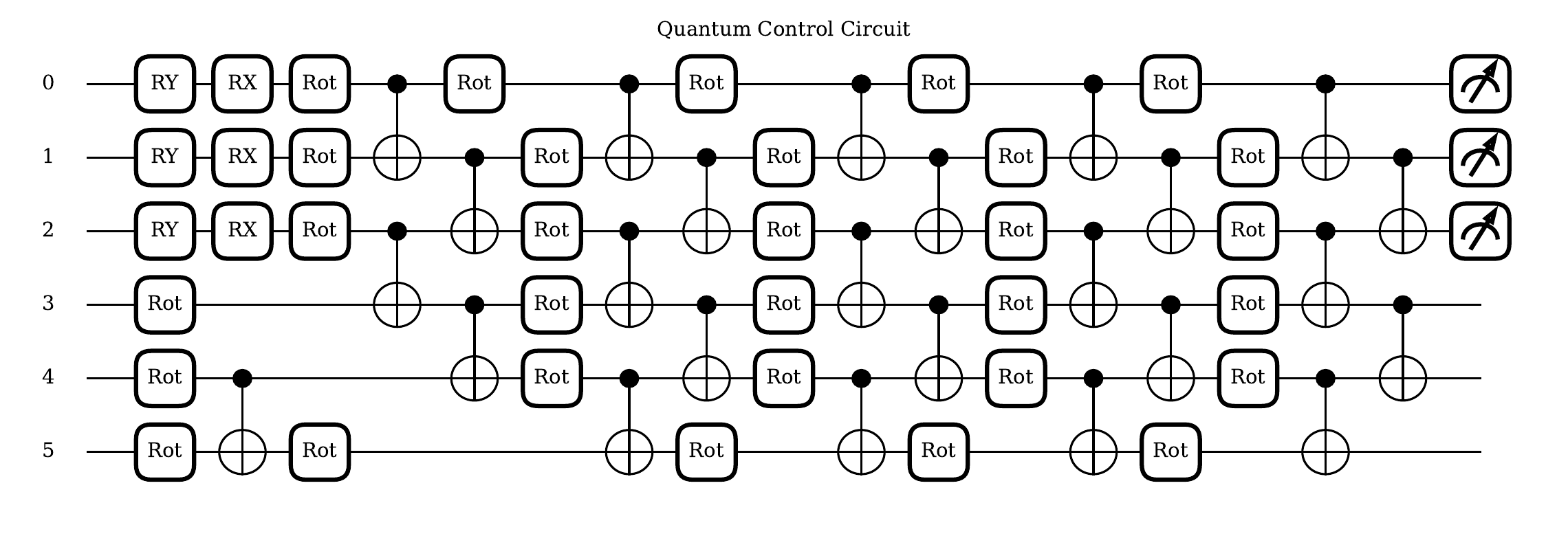}
\caption{Quantum circuit simulating the optimal choice of controls for the energy efficient climate control model.}
\label{fig11}
\end{figure}

\begin{figure}[H]
\centering
\includegraphics[width=1\linewidth]{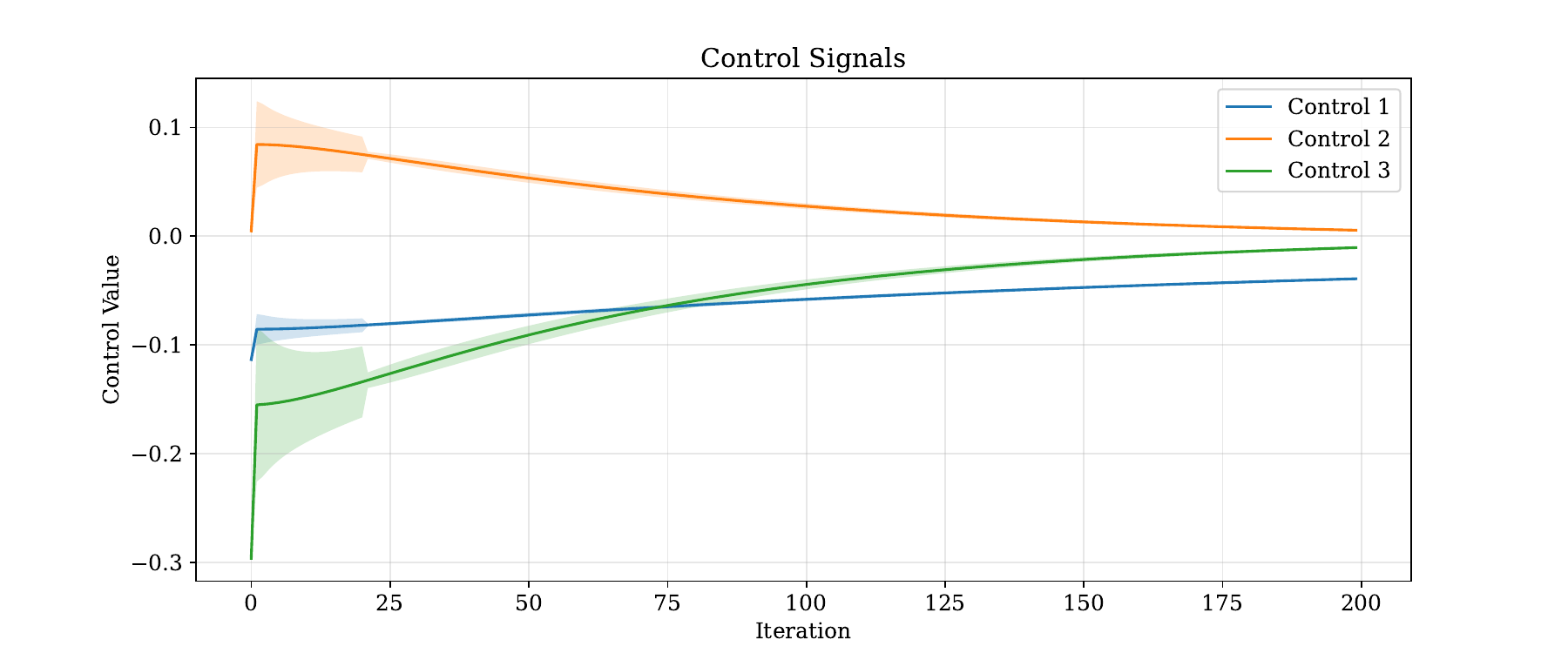}
\caption{Optimal control response variable (in this case, the power delivered by the HVAC system) for the energy efficient climate control model.}
\label{fig12}
\end{figure}

\begin{figure}[H]
\centering
\includegraphics[width=1\linewidth]{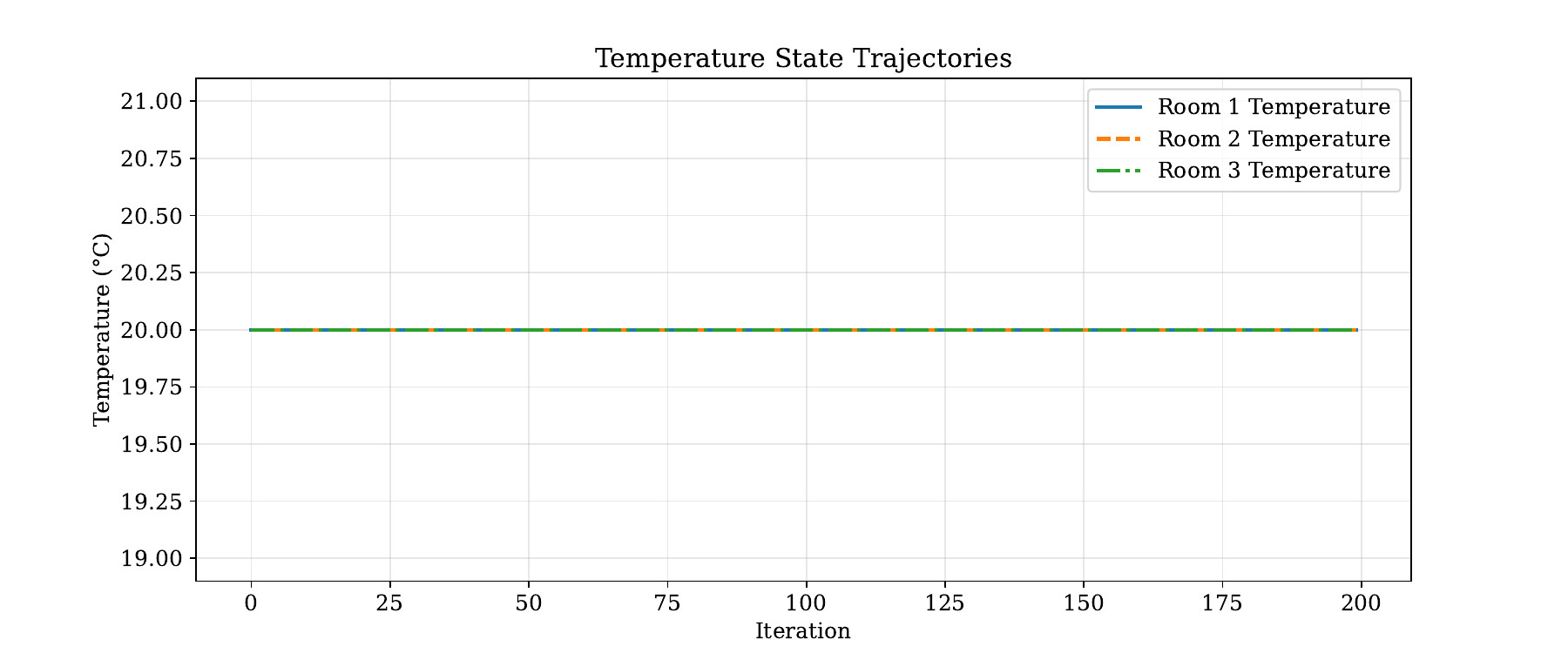}
\caption{Optimal trajectories for the state variable $T_{\text{room},k}$ for the energy efficient climate control system.}
\label{fig13}
\end{figure}

\begin{figure}[H]
\centering
\includegraphics[width=1\linewidth]{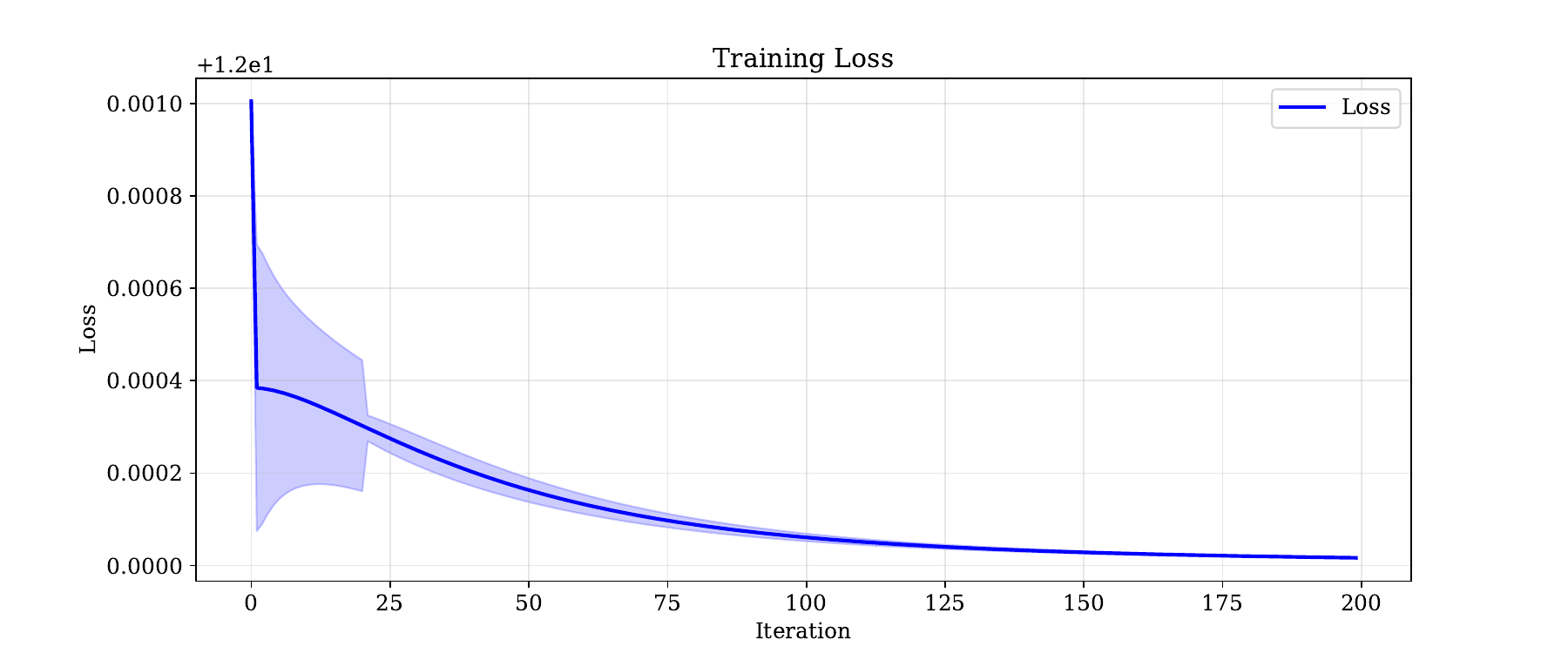}
\caption{Loss function for the energy efficient climate control model.}
\label{fig14}
\end{figure}

From \ref{fig11}, we use a quantum circuit consisting of 6 qubits, with state preparation occurring at qubits 0, 1, and 2. The parametrized rotations are used to learn the optimal parameters, and the CNOT gate creates entanglement. Measurement occurs at qubits 0, 1, and 2 in the computational basis. 

In Fig. \ref{fig12}, we see that all three control signals exhibit significant fluctuations in the early iterations, approximately within the first 25 iterations. This is characteristic of the algorithm exploring the search space in the presence of noise and the complex cost landscapes in \eqref{loss for room heating}. Further, this is also potentially attributed to the quantum circuit components. The shaded regions are wider during this phase, which indicates higher uncertainty in the control signal values. As the iterations progress, the fluctuations dampen, and the control signals appear to converge towards more stable values. This suggests that the optimization algorithm finds a region of the parameter space that yields a lower cost. For control $u_{1}$ (blue), the signal starts with negative values. It gradually increases towards a plateau around $-0.05$, and the uncertainty associated with $u_{1}$ decreases significantly after the initial phase, implying a more confident determination of its optimal value. For control $u_{2}$ (orange), the signal starts with a value close to 0, rapidly increases to a peak around 0.1, and then gradually decreases towards a value close to zero by the end of the 200 iterations. The uncertainty in $u_{2}$ also decreases over time, indicating increasing confidence in its trajectory. The initial spike may be necessary to steer the system towards a desired state before tending towards an equilibrium. For $u_{3}$, the signal starts with a relatively large negative value and steadily increases towards a value slightly below 0. The uncertainty in $u_{3}$ also diminishes as the optimization progresses. $u_{3}$ appears to play a role in counteracting some initial deviation or driving the system in a specific direction. 

From Fig. \ref{fig13}, we observe that there is no discernible change or fluctuation in the temperature of any of the rooms. All three temperature trajectories are identical, as evidenced by the feint overlaps of their trajectories, and remain constant at $20^{\circ}\text{C}$. We attribute this constant value for all three state trajectories to either of three reasons:
\begin{enumerate}
\item The simulation might have started with all rooms already at the desired temperature of $20^{\circ}\text{C}$. If the control objective is to maintain this temperature and the system is well-behaved, the MPC controller would ideally apply minimal to no control actions, resulting in constant temperatures.
\item Through the utilization of the quantum properties, the MPC controller could have very quickly identified and applied the optimal control actions in the initial iterations to bring the room temperatures to the set point temperature of $20^{\circ}\text{C}$. Furthermore, these control actions are maintained throughout the simulation, which indicates that the system is in a stable state where no further adjustments are needed. 
\item The model of the thermal dynamics of the rooms used within the MPC might be too simplistic. 
\end{enumerate}
We believe that a combination of these factors resulted in this trend that we observed. Potentially, with more realistic modeling of the room temperature changes, the state trajectory variable would exhibit more expected profiles.

In Fig. \ref{fig14}, we see that the training loss starts at a relatively high value and decreases sharply during the first 25 iterations. This indicates that the algorithm makes significant progress in finding better parameter values for the quantum components, leading to a substantial improvement in the controller's performance. After the initial rapid decrease, the loss continues to decrease but at a slower rate. This is typical of optimization algorithms as they approach a local minimum, and the smaller gradients in this region lead to smaller updates in the parameters. The shaded blue area around the loss curve, representing the variance in the loss, is also relatively large in the initial iterations but decreases significantly as training progresses. This suggests that the optimization process becomes more stable and converges towards a consistent region of the parameter space. The initial high variance is due to two reasons: The stochastic nature of quantum measurements and the exploration phase of the algorithm.

\subsection{Experiment 3: Autonomous Vehicular Dynamics on a Curvilinear Road}\label{autonomous vehicle experiment}
For an autonomous vehicle navigating a curvilinear road, we consider both longitudinal and lateral dynamics. We define $\mathbf{x}=\begin{pmatrix}x_{1} &x_{2} &x_{3} &x_{4}\end{pmatrix}^{T}$ with $x_{1}\overset{.}{=}s$ being the longitudinal position along the road, $x_{2}\overset{.}{=}v$ being the longitudinal velocity, $x_{3}\overset{.}{=}y$ being the lateral deviation from the centerline of the road, and $x_{4}\overset{.}{=}\vartheta$ is the heading angle relative to the road tangent. We define the control inputs to be $\mathbf{u}=\begin{pmatrix}u_{1} &u_{2}\end{pmatrix}^{T}$ with $u_{1}$ being the traction/force required for braking, and $u_{2}$ being the steering angle. Using a curvilinear frame of reference, the dynamics of the autonomous vehicle can be described by the nonlinear coupled system of ODEs
\begin{equation}\label{vehicle dynamics equation}
\left\{
\begin{aligned}
\dot{x}_{1}=&\;x_{2}\cos x_{4}, \\
\dot{x}_{2}=&\;\frac{1}{m}\left(u_{1}-F_{\text{drag}}-F_{\text{rolling}}\right), \\
\dot{x}_{3}=&\;x_{2}\sin x_{4}, \\ 
\dot{x}_{4}=&\;\frac{1}{L}x_{2}\tan u_{2}-\kappa(x_{1})x_{2},
\end{aligned}
\right. 
\end{equation}
where $m$ is the mass of the vehicle, $L$ is the length of the wheelbase, $F_{\text{drag}}=\frac{1}{2}C_{d}\rho Ax_{2}^{2}$ is the drag force with $C_{d}$ being the drag coefficient, $\rho$ is the density of air, $A$ is the cross-sectional area, $F_{\text{rolling}}=C_{r}mg$ with $C_{r}$ being the force of resistance due to rolling and $C_{r}$ rolling friction coefficient, and $\kappa(\cdot)$ being the curvature of the road. We note that the equations in \eqref{vehicle dynamics equation} are a gross simplification of the models derived in \cite{kong2015kinematic,xu2019design}. Using Euler's method, we can discretize \eqref{vehicle dynamics equation} to obtain
\begin{equation}
\left\{
\begin{aligned}
x_{1}(k+1)=&\;x_{1}(k)+x_{2}(k)\cos[x_{4}(k)]\Delta t, \\
x_{2}(k+1)=&\;x_{2}(k)+\frac{1}{m}\left[u_{1}(k)-\frac{1}{2}C_{d}\rho Ax_{2}(k)^{2}-C_{r}mg\right]\Delta t, \\
x_{3}(k+1)=&\;x_{3}(k)+x_{2}(k)\sin[x_{4}(k)]\Delta t, \\
x_{4}(k+1)=&\;x_{4}(k)+\left\{\frac{1}{L}x_{2}(k)\tan[u_{2}(k)]-\kappa[x_{1}(k)]x_{2}(k)\right\}\Delta t.
\end{aligned}
\right. 
\end{equation}
Our goal is to minimize the loss function 
\begin{equation}\label{autonomous vehicle loss}
\begin{aligned}
J=&\;\sum_{k=0}^{N}\left(||x_{3}(k)||^{2}+||x_{4}(k)||^{2}+\lambda_{1}||u_{1}(k)||^{2}+\lambda_{2}||u_{2}(k)||^{2}\right.\\
&\left.+\lambda_{3}||\Delta u_{1}(k)||^{2}+\lambda_{4}||\Delta u_{2}(k)||^{2}\right),
\end{aligned}
\end{equation}
with $0\leq\lambda_{1},\lambda_{2},\lambda_{3},\lambda_{4}\leq 1$ being the penalization weights.  

We perform the experiment with the following parameter values: $m=1500\;\text{kg}, L=2.5\;\text{m}, \Delta t=0.1, g=9.81\;\text{m/s}^{2},C_{d}=0.3, 1.225\;\text{kg/m}^{3}\;\text{density of air},A=2.5\;\text{m}^{2}\;\text{(frontal area)},C_{r}=0.01,\kappa=0,\lambda_{1}=\lambda_{2}=0.1, \lambda_{3}=\lambda_{4}=0.01$. 

\begin{figure}[H]
\centering
\includegraphics[width=1.2\linewidth]{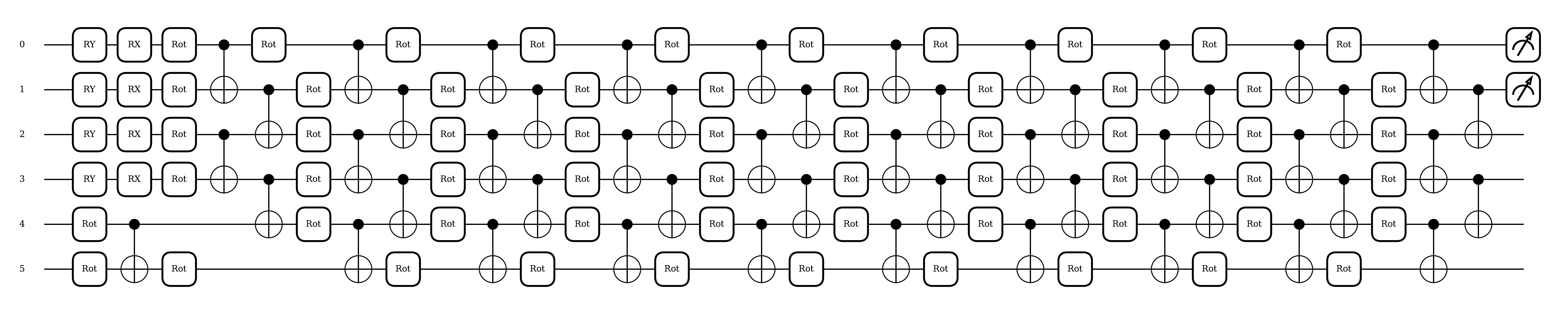}
\caption{Quantum circuit simulating the optimal choice of controls for the autonomous vehicle model.}
\label{fig15}
\end{figure}

\begin{figure}[H]
\centering
\includegraphics[width=1\linewidth]{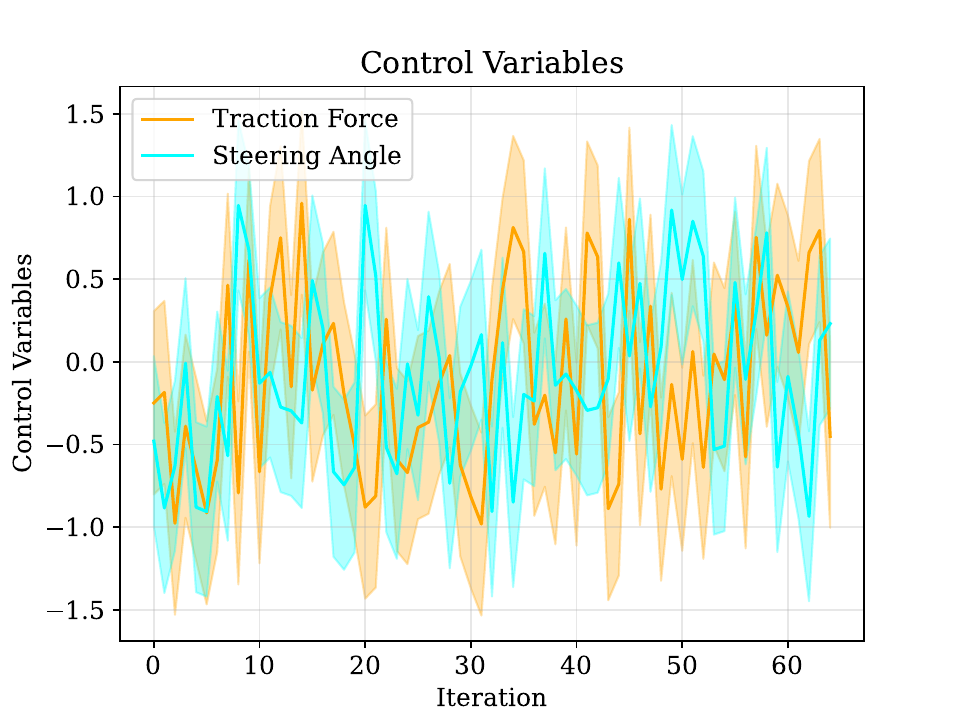}
\caption{Optimal control response variables for the autonomous vehicle model.}
\label{fig16}
\end{figure}

\begin{figure}[H]
\centering
\includegraphics[width=1\linewidth]{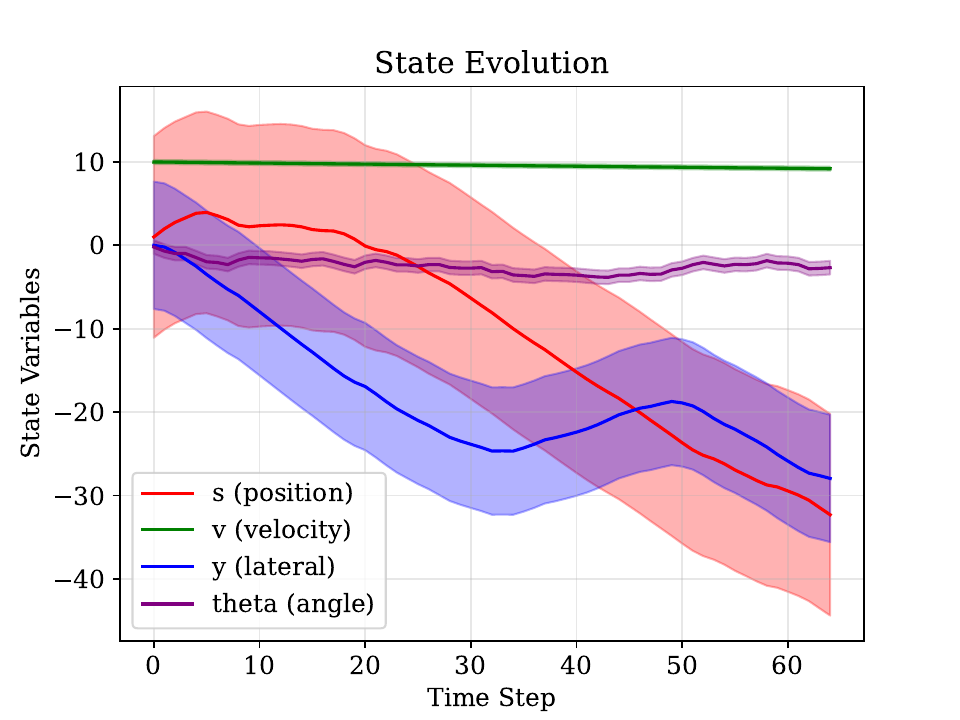}
\caption{Optimal trajectories for the state variables for the  autonomous vehicle model.}
\label{fig17}
\end{figure}

\begin{figure}[H]
\centering
\includegraphics[width=1\linewidth]{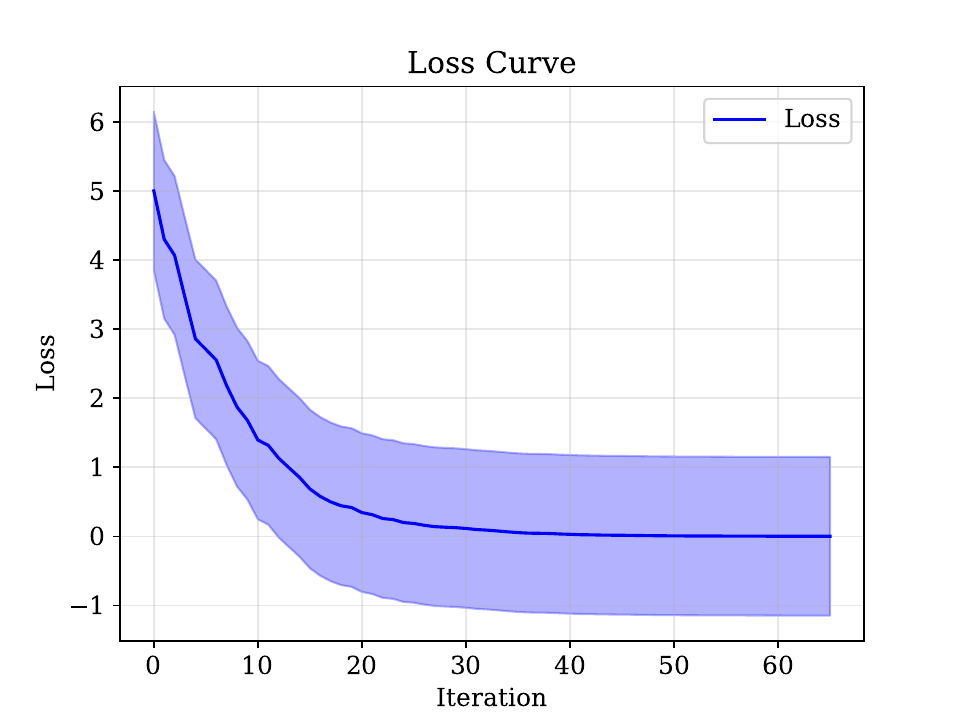}
\caption{Loss function for the  autonomous vehicle model.}
\label{fig18}
\end{figure}

As seen in Fig. \ref{fig15}, we used 6 qubits to determine the optimal controls, and the setup is analogous to previous experiments in how encoding, parametrized learning, entanglement, and measurement. 

In Fig. \ref{fig16}, we see the control variables have very erratic behavior. Both the traction force (oscillating in $\left[-1.5,1.5\right]$) and steering angle (oscillating in $\left[-1,1\right]$) exhibit significant oscillations throughout the 60 iterations. This suggests that the optimization algorithm is continuously adjusting these control variables, exploring the complex cost landscape, and reacting to uncertainties in the system. The system does not appear to settle into a steady state within the observed iterations. The shaded region around the traction force indicates a considerable level of uncertainty associated with this control variable, especially in the earlier iterations, and thereafter a slight dampening occurs. Similar to the traction force, the steering angle exhibits substantial uncertainty, particularly in the initial phase. The oscillations seem to persist throughout the iterations, suggesting that finding a stable steering strategy is challenging.

In Fig. \ref{fig17}, we observe that the position of the autonomous vehicle starts near $0$ and generally decreases over time, reaching approximately $-30$ by the end of the simulation. Further, there is significant uncertainty in the position, as indicated by the wide red-shaded variance region around the $s$ trajectory, especially in the middle portion of the simulation. This suggests that the prediction of the position is subject to considerable variability, possibly due to uncertainties in the system dynamics and control. The velocity starts near $10$ and remains relatively constant with a slight downward trend towards the end, staying around $8$. The uncertainty in velocity (green shaded region) is relatively small compared to position, indicating a more predictable behavior for this state variable. The almost constant positive velocity suggests a general movement in one direction. The lateral displacement starts near 0, decreases significantly to around $-30$ by iteration $30$, and then starts to increase again, reaching around $-20$ by the end. The uncertainty in lateral displacement (blue-shaded region) is also substantial, suggesting challenges in accurately predicting or controlling the lateral movement. The trajectory indicates an initial drift to one side followed by a partial correction. The angle starts near $0$, decreases slightly to around $-5$, and then fluctuates within a relatively narrow range between $-5$ and $-2$. The uncertainty in the angle (purple-shaded region) is relatively small after the initial time steps, indicating that the angular orientation is being controlled with reasonable precision.

The loss curve in Fig. \ref{fig18} starts at a relatively high value (around $6$) and decreases sharply during the first $15$ iterations. This indicates that the algorithm makes significant progress in finding better parameter values for the controls, leading to a substantial improvement in the controller's performance. The shaded blue area, representing the variance in the loss, is also relatively large in the initial iterations but decreases significantly as training progresses. This suggests that the optimization process becomes more stable and converges towards a consistent region of the parameter space. We attribute the initial high variance to the stochastic nature of quantum measurements in the algorithm's exploration phase.

\subsection{Experiment 4: The Simple Pendulum}\label{experiment2}
The classical pendulum, as seen in Fig. \ref{fig5}, consists of a mass $m$ fixed at one end, tied to an inextensible string of length $\ell$, making an angle of $\Theta$ with the vertical. It can easily be verified that the dynamics of this system is governed by
\begin{equation}\label{pendulum eqn}
\ddot{\Theta}=-\frac{g}{\ell}\sin\Theta+\frac{u}{m\ell^{2}},
\end{equation}
where $\ddot{\Theta}$ is the angular acceleration of the pendulum, $g\approx 9.81\;\text{m/s}^{2}$ is the acceleration due to gravity, and $u$ is the torque which is taken to be the control variable. Using Euler's method, we can easily discretize \eqref{pendulum eqn} into the form
\begin{equation}\label{discrete simple pendulum equation}
\left\{
\begin{aligned}
\Theta_{t+1}=&\;\Theta_{t}+\dot{\Theta}_{t}\Delta t, \\
\dot{\Theta}_{t+1}=&\;\dot{\Theta}_{t}+\left(-\frac{g}{\ell}\sin\Theta_{t}+\frac{u_{t}}{m\ell^{2}}\right)\Delta t,
\end{aligned}
\right.
\end{equation}
where $\dot{\Theta}$ denotes the angular velocity, and $\Delta t$ is the discrete time step.

\begin{figure}[H]
\centering
\includegraphics[width=1\linewidth]{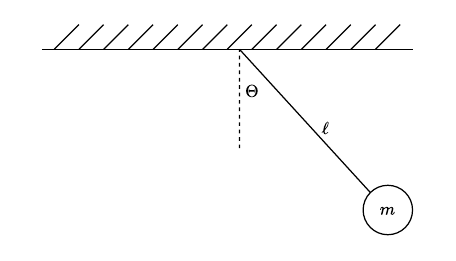}
\caption{The classical pendulum problem, fixed at one end and oscillating in a circular path.}
\label{fig5}
\end{figure}

For simplicity in training the model, we set $\ell=1, m=1, g=9.81$ and trained the model over 50 iterations. The other training parameters chosen are $\lambda=0.05,u_{\min}=-2,u_{\max}=2$, and we use a dynamical learning rate so as to ensure that the oscillatory behavior of the training loss is minimized. Specifically, in order to update the gradients, we have used SGD with momentum and chose an initial learning rate of 0.3, a minimum learning rate of 0.01, and a decay rate of 0.95. We set the momentum to be 0.85, and we clipped the gradient between $\left[-0.5,0.5\right]$. The results of the training are depicted in Figs. \ref{fig6}-\ref{fig9}. 

\begin{figure}[H]
\centering
\includegraphics[width=1\linewidth]{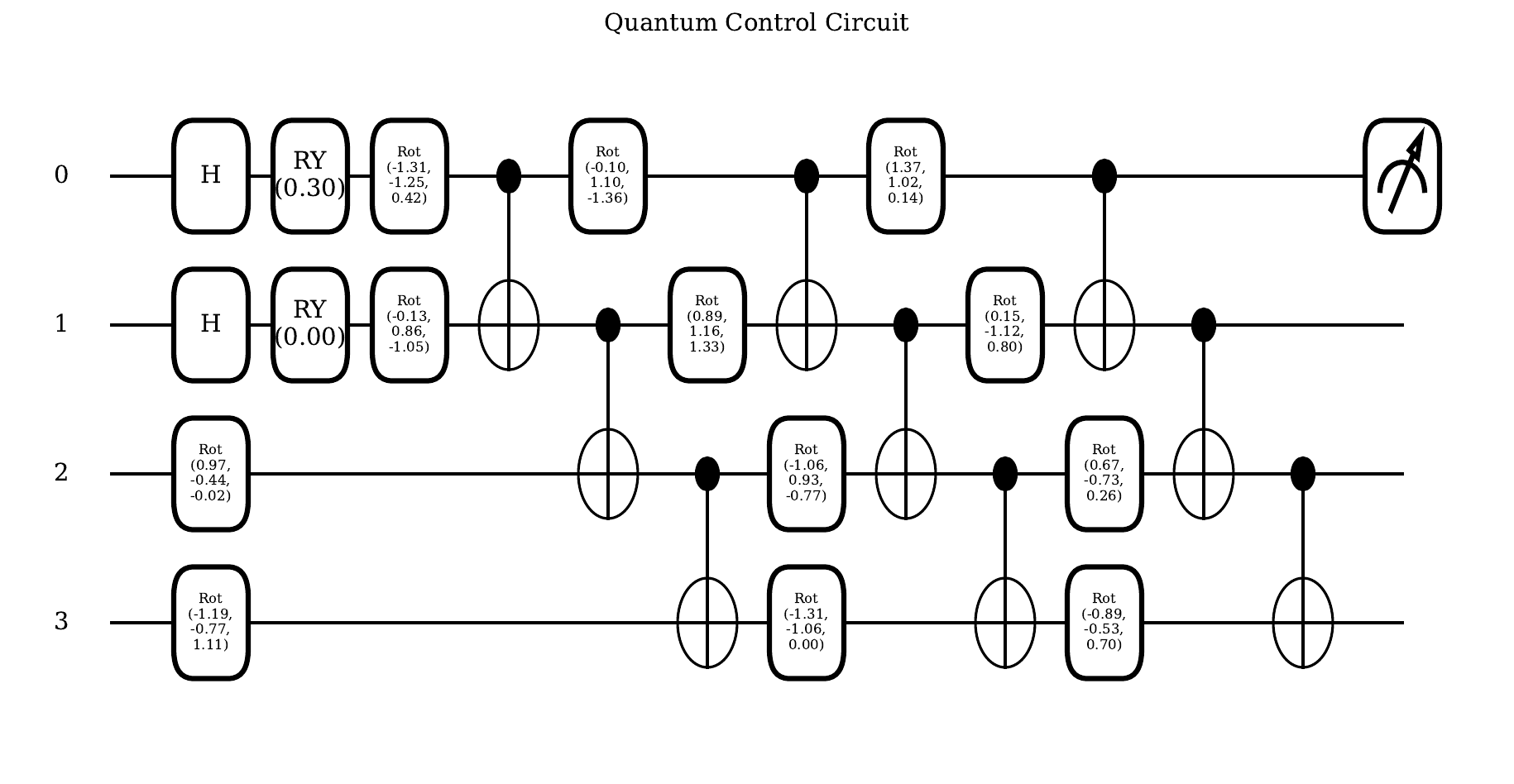}
\caption{Quantum circuit simulating the optimal choice of controls for the simple pendulum system.}
\label{fig6}
\end{figure}

\begin{figure}[H]
\centering
\includegraphics[width=1\linewidth]{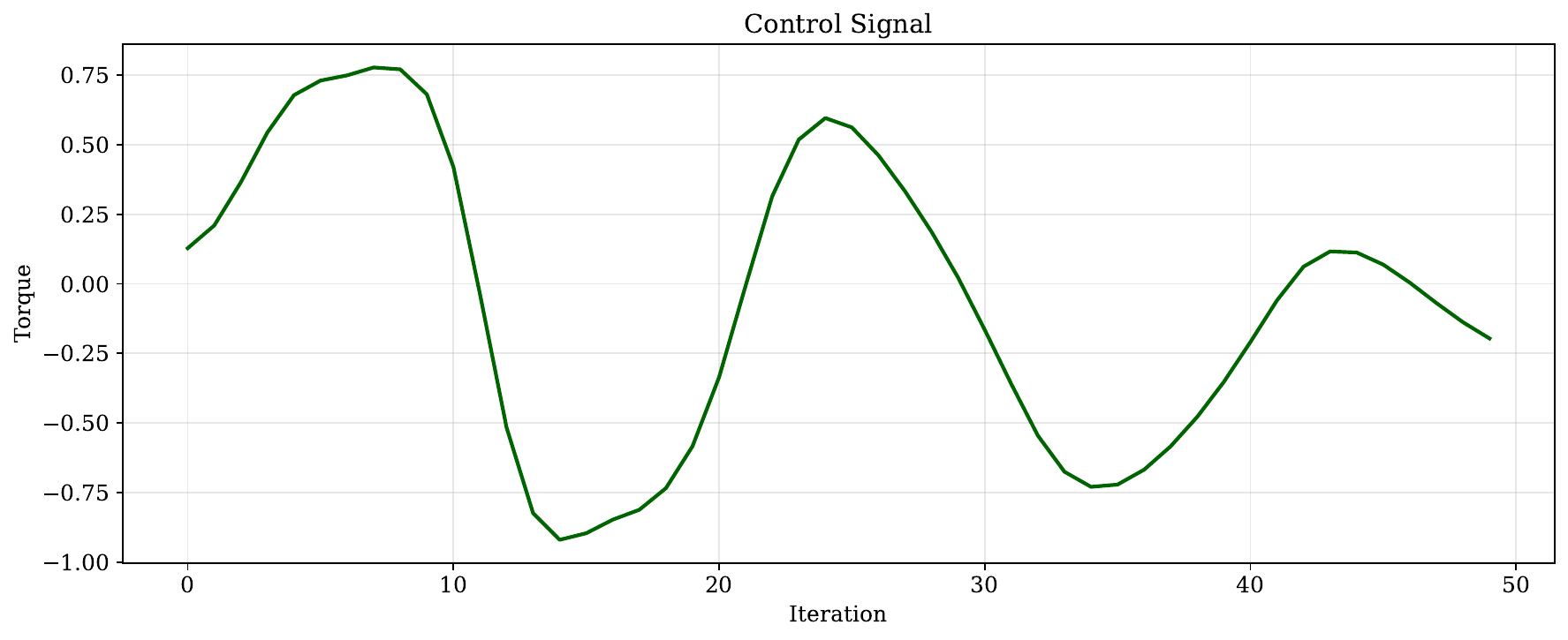}
\caption{Optimal control response variable (in this case, the toque) for the simple pendulum system.}
\label{fig7}
\end{figure}

\begin{figure}[H]
\centering
\includegraphics[width=1\linewidth]{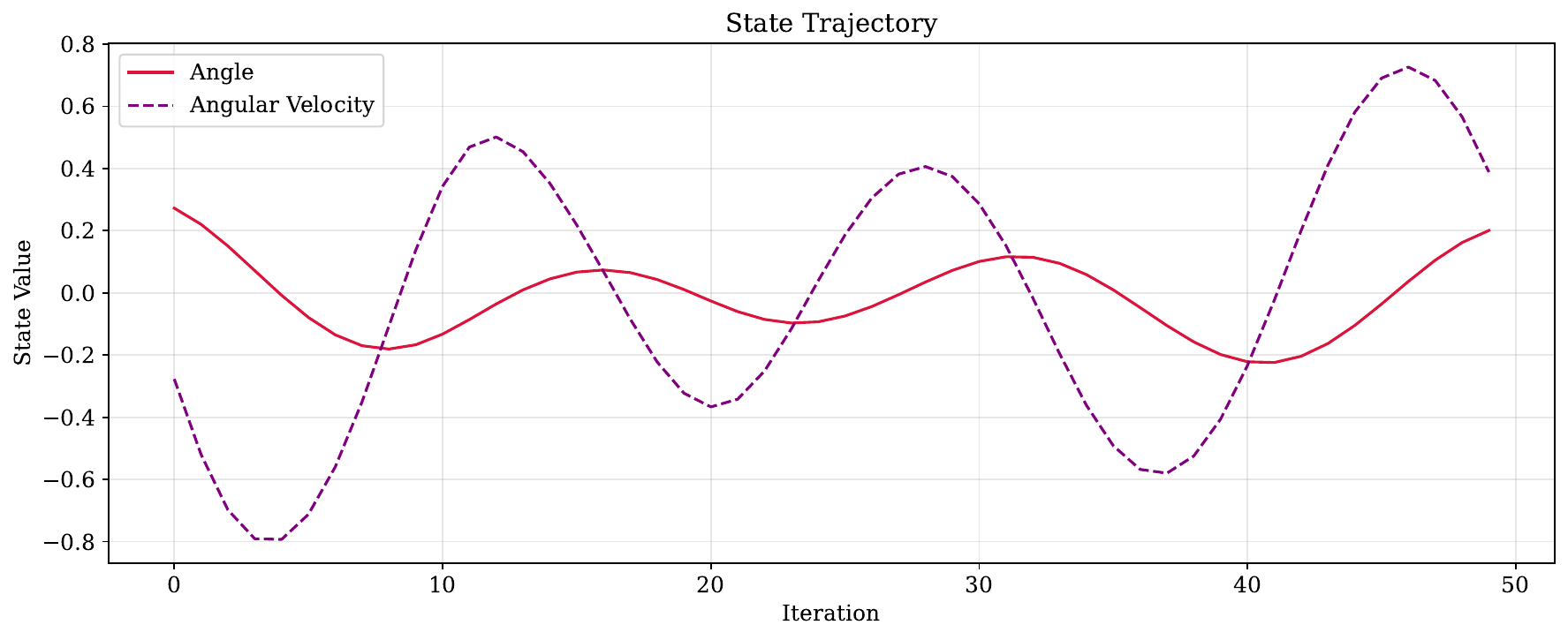}
\caption{Optimal trajectories for the state variables $\Theta$ and $\dot{\Theta}$ for the simple pendulum system.}
\label{fig8}
\end{figure}

\begin{figure}[H]
\centering
\includegraphics[width=1\linewidth]{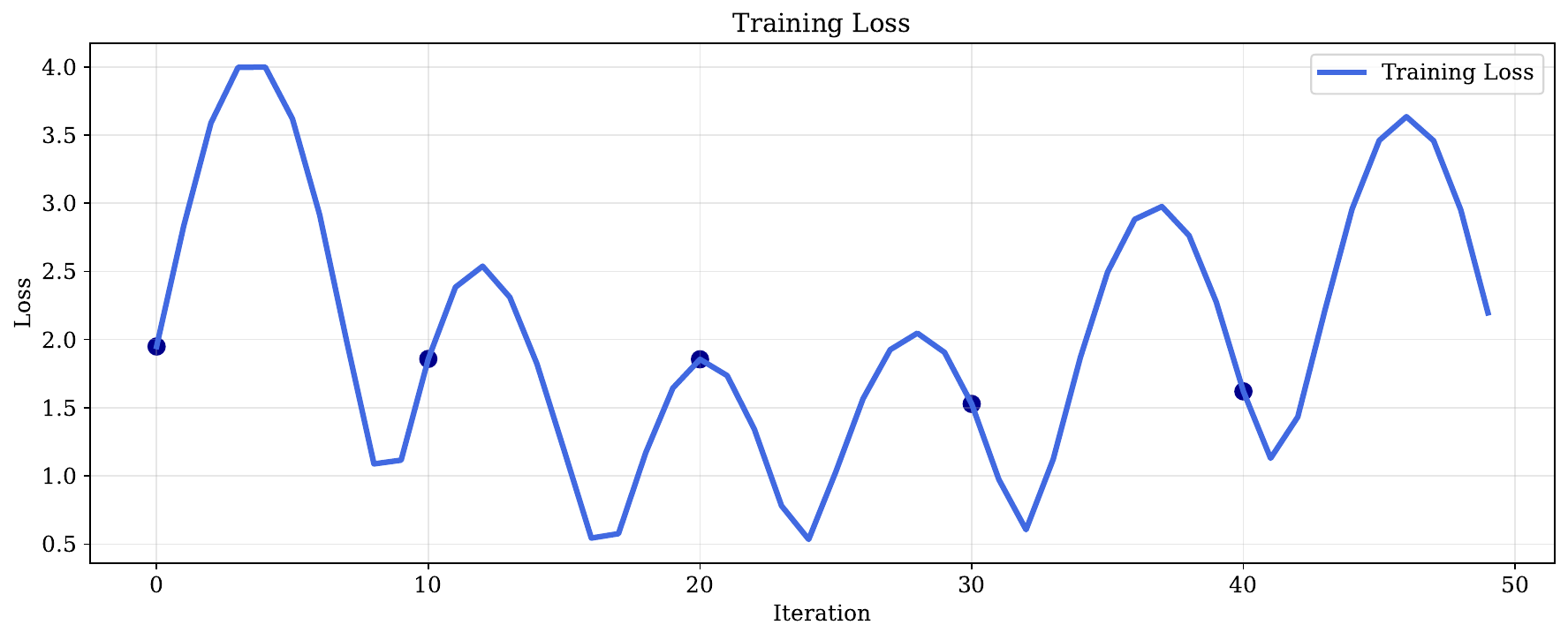}
\caption{Loss function for the simple pendulum system, with the loss being tracked every 10 iterations.}
\label{fig9}
\end{figure}

In Fig. \ref{fig6}, the classical state variables in \eqref{discrete simple pendulum equation} $\left(\Theta,\dot{\Theta}\right)$ are mapped into quantum states according to
\begin{equation}
\mathbf{x}_{0}=\left(\Theta,\dot{\Theta}\right)\to H\otimes H\ket{00},
\end{equation}
where $H$ is the Hadarmard gate used to create superposition. Simultaneously, rotation gates are applied to encode the classical states into quantum amplitudes according to
\begin{equation}
\ket{\psi(\Theta,\dot{\Theta})}=RY(\Theta)\otimes RY(\dot{\Theta})\ket{00}. 
\end{equation}
Similar to Sec. \ref{experiment2}, the rotation gates are applied in order to encode the control policy in a quantum variational form and discover the trainable parameters by undergoing single-qubit rotations. Thereafter, entanglement is introduced between qubits using CNOT gates, and finally, measurement is performed in the $Z$-basis in order to obtain the control torque. 

In Fig. \ref{fig7}, we see the representation of the optimal control signal evolution, with the control variable being torque, over iterations. In an ideal MPC setup, the control signal should stabilize over time as the system reaches equilibrium. Here, the control signal exhibits persistent oscillations, which is indicative that the quantum circuit parameters used, such as choice of gates and circuit depth, may be suboptimal, and there is a high variance in the expected values extracted from the circuit. We attribute this to the oversensitivity of the quantum controller to small changes in the state, which leads to the oscillatory behavior -- analogous to the exploding gradients problem in ML, where the updates to the learning algorithm are too aggressive.  

In Fig. \ref{fig8}, we observe the trajectories of the simple pendulum system with the solid line (red) representing the position of the pendulum $\Theta$ and the dashed line (purple) representing the angular velocity $\dot{\Theta}$. In an ideal control scenario, the pendulum should stabilize near the target state $\Theta=0,\dot{\Theta}=0$ -- that is when it settles in an upright or downward position with minimal oscillations. However, the presence of persistent oscillations in the diagram suggests potential instabilities in the control policy, and there is no effective damping of the motion. We account for this, as with the discussion of Fig. \ref{fig7}, by attributing this behavior to The control actions generated by the quantum circuit having high variance, leading to overcorrections and delayed adjustments.

In Fig. \ref{fig9}, we observe that the loss function fluctuates significantly, showing a non-monotonic trend. Further, there are multiple peaks and valleys, indicating that the training process encounters many local minima. The loss does not consistently decrease, which suggests possible issues with gradient updates and optimization dynamics. The initial loss at step 0 was 1.9489, and after 50 iterations, it decreased to 1.6196. While technically, the loss function did decrease, it was not a consistent decrease. 

\subsection{Experiment 5: The Compound Pendulum}\label{experiment3}
We consider the double pendulum system depicted in Fig. \ref{fig10}, with mass $m_{1}$ tied to a string of length $\ell_{1}$ which is tied to a pivot point, making an angle $\Theta$ with the vertical. Another mass $m_{2}$ is tied to $m_{1}$ via a string of length $\ell_{2}$ and makes an angle $\varphi$ with the vertical.

\begin{figure}[H]
\centering
\includegraphics[width=1.0\linewidth]{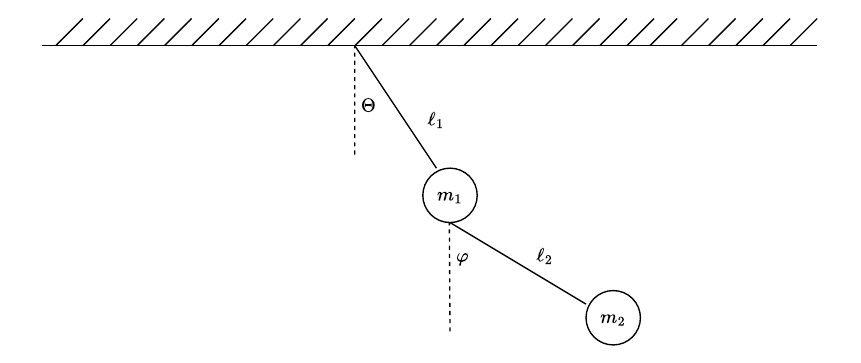}
\caption{The compound pendulum consists of two differing masses tied to strings of differing lengths, subtending angles of different sizes.}
\label{fig10}
\end{figure}

It can be shown using Lagrangian mechanics that the second-order nonlinear coupled ODE system gives the equations of motion that determine the system, ignoring frictional effects, is
\begin{equation}\label{double pendulum equations}
\left\{
\begin{aligned}
&\left(m_{1}+m_{2}\right)\ell_{1}^{2}\ddot{\Theta}+m_{2}\ell_{1}\ell_{2}\cos(\Theta-\varphi)+m_{2}\ell_{1}\ell_{2}\sin(\Theta-\varphi) \\
&+\left(m_{1}+m_{2}\right)g\ell_{1}\sin\Theta=\tau_{\Theta}, \\
&m_{2}\ell_{2}^{2}\ddot{\varphi}+m_{2}\ell_{1}\ell_{2}\ddot{\Theta}\cos(\Theta-\varphi)-m_{2}\ell_{1}\ell_{2}\dot{\Theta}^{2}\sin(\Theta-\varphi)+m_{2}g\ell_{2}\sin\varphi=\tau_{\varphi},
\end{aligned}
\right.
\end{equation}
where $\tau_{\Theta}$ and $\tau_{\varphi}$ are the torques applied at the pivot points where the strings are tied, and serve as controls. It can easily be verified that the Euler discretization of the system \eqref{double pendulum equations} is
\begin{equation}
\mathbf{x}_{k+1}=f(\mathbf{x}_{k},\mathbf{u}_{k})=
\begin{pmatrix}
\Theta_{k}+\dot{\Theta}_{k}\Delta t \\
\dot{\Theta}_{k}+\ddot{\Theta}_{k}\Delta t \\
\varphi_{k}+\dot{\varphi}_{k}\Delta t \\
\dot{\varphi}_{k}+\ddot{\varphi}_{k}\Delta t
\end{pmatrix},\quad 
\begin{pmatrix}
\ddot{\Theta}_{k} \\
\ddot{\varphi}_{k}
\end{pmatrix}=
\mathbf{M}(\mathbf{x}_{k})^{-1}\cdot\left[\boldsymbol{\tau}_{k}-\mathbf{C}(\mathbf{x}_{k})\right],
\end{equation}
with 
\begin{equation}
\begin{aligned}
\mathbf{M}(\mathbf{x})=&\;
\begin{pmatrix}
\left(m_{1}+m_{2}\right)\ell_{1}^{2} &m_{2}\ell_{1}\ell_{2}\cos(\Theta-\varphi) \\
m_{2}\ell_{1}\ell_{2}\cos(\Theta-\varphi) &m_{2}\ell_{2}^{2}
\end{pmatrix}, \\
\mathbf{C}(\mathbf{x})=&\;
\begin{pmatrix}
m_{2}\ell_{1}\ell_{2}\sin(\Theta-\varphi)\dot{\varphi}^{2}+\left(m_{1}+m_{2}\right)g\ell_{1}\sin\Theta \\
-m_{2}\ell_{1}\ell_{2}\sin(\Theta-\varphi)\dot{\Theta}^{2}+m_{2}g\ell_{2}\sin\varphi
\end{pmatrix}, \\
\boldsymbol{\tau}=&\;
\begin{pmatrix}
\tau_{\Theta} \\
\tau_{\varphi}
\end{pmatrix}, 
\end{aligned}
\end{equation}
where $\mathbf{M}$ is the mass matrix that defines the relationship between the system's accelerations and forces., $\mathbf{C}$ is the Coriolis vector that accounts for the apparent deflection of the system when viewed from a rotating frame as a result of their motion and the gravitational force, and $\boldsymbol{\tau}$ is the torque/control vector.  

We minimize the loss function 
\begin{equation}\label{double pendulum loss}
J=\sum_{k=0}^{N-1}\left(\lambda_{1}||\mathbf{x}_{k}-\mathbf{x}_{\text{target}}||^{2}+\lambda_{2}||\mathbf{u}_{k}||^{2}+\lambda_{3}||\boldsymbol{\theta}_{k}||^{2}\right),
\end{equation}
where $0\leq\lambda_{1},\lambda_{2},\lambda_{3}\leq 1$ are the penalization weights, as in the other experiments. Specifically, $\lambda_{1}$ controls the state error term and guides the quantum optimizer to minimize the difference between the current and target states over iterations. The $\lambda_{2}$ term encourages energy-efficient and physically realistic control actions to avoid aggressive and unstable behavior. The $\lambda_{3}$ term regularizes the VQC to prevent excessively large parameter updates and assists in maintaining stability and generalization in the quantum control policy by avoiding overfitting to noise. The ``new'' term is the loss due to the quantum network by summing up the learnable parameter values. We chose this form of the loss function because, upon experimentation, we discovered that the loss curve was repeatedly giving an increasing profile, both monotonic and erratic, and we were able to make the loss function decrease, albeit slightly, by choosing the loss to take on this form.  

We simulate the system in \eqref{double pendulum equations} using the parameter values: $m_{1}=m_{2}=1\;\text{kg}, \ell_{1}=\ell_{2}=1\;\text{m}, g=9.81\;\text{m/s}^{2},\Delta t=0.05\;\text{s},\lambda_{1}=1, \lambda_{2}=0.1,\lambda_{3}=0.01$. Our initial states are $\mathbf{x}_{0}=\begin{pmatrix}\Theta &\dot{\Theta} &\varphi &\dot{\varphi}\end{pmatrix}^{T}=\begin{pmatrix}0.1 &0 &0.1 &0 \end{pmatrix}^{T}$ and our target states are $\mathbf{x}_{\text{target}}=\begin{pmatrix}0.79 &0 &0.52 &0 \end{pmatrix}^{T}$. The results are summarized in Figs.  

\begin{figure}[H]
\centering
\includegraphics[width=1.0\linewidth]{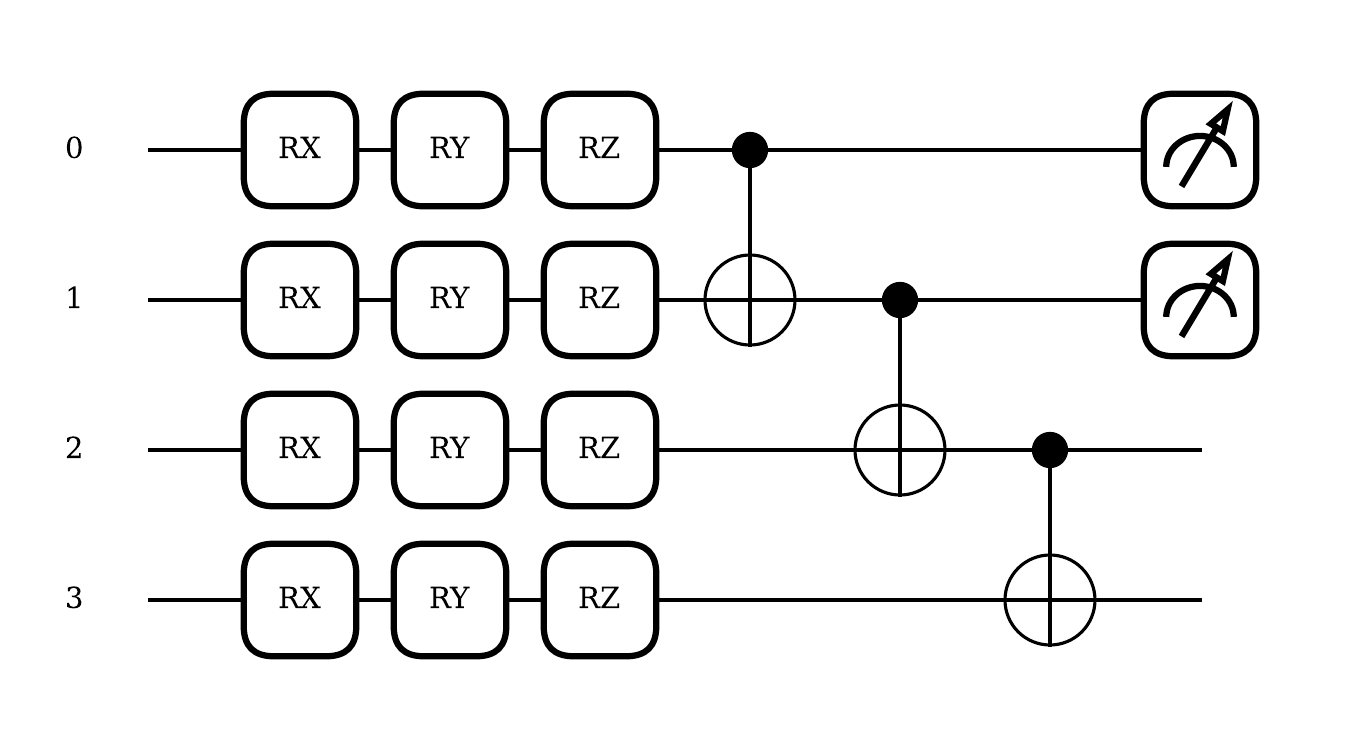}
\caption{Quantum circuit simulating the optimal choice of controls for the double pendulum.}
\label{fig19}
\end{figure}

\begin{figure}[H]
\centering
\includegraphics[width=1\linewidth]{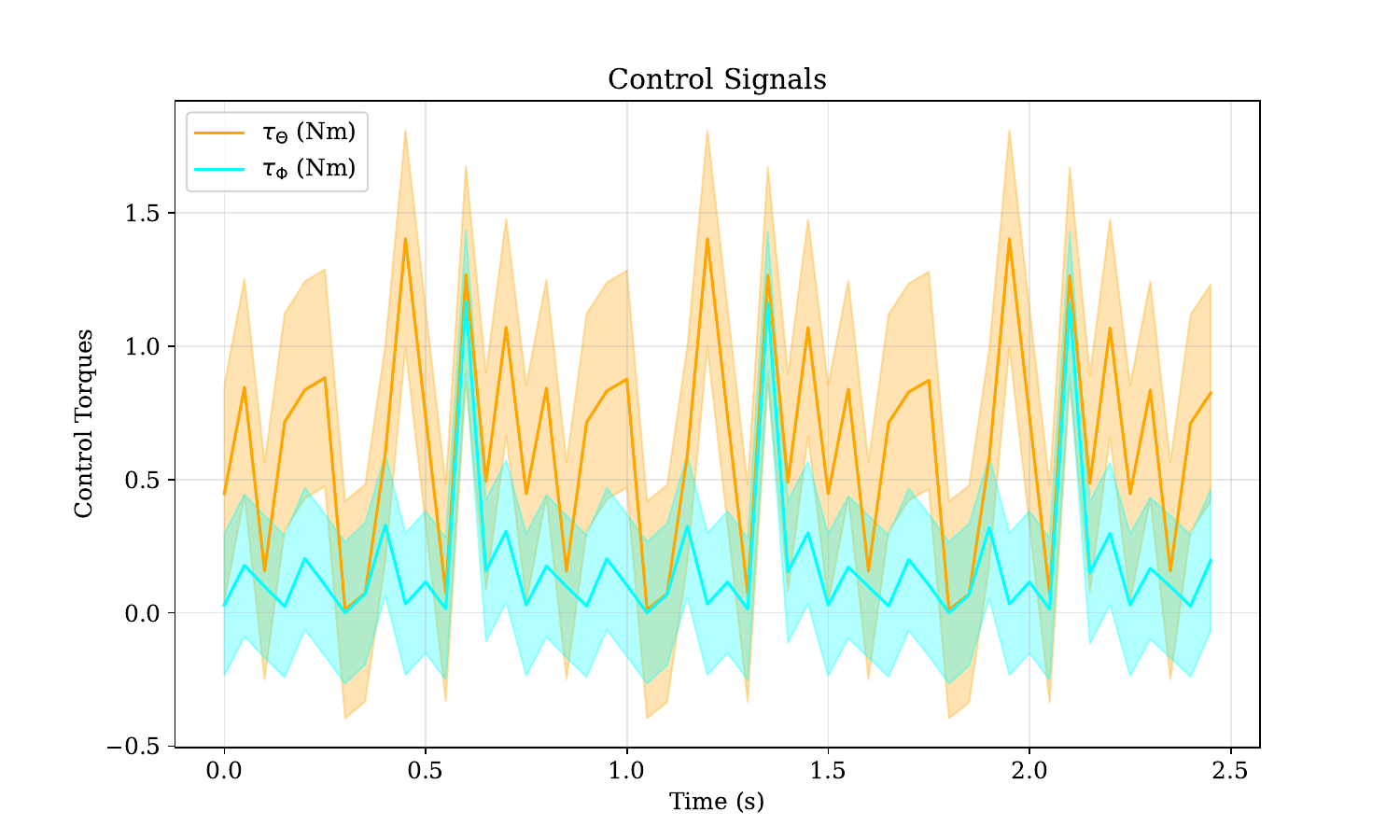}
\caption{Optimal control response variables for the double pendulum.}
\label{fig20}
\end{figure}

\begin{figure}[H]
\centering
\includegraphics[width=1\linewidth]{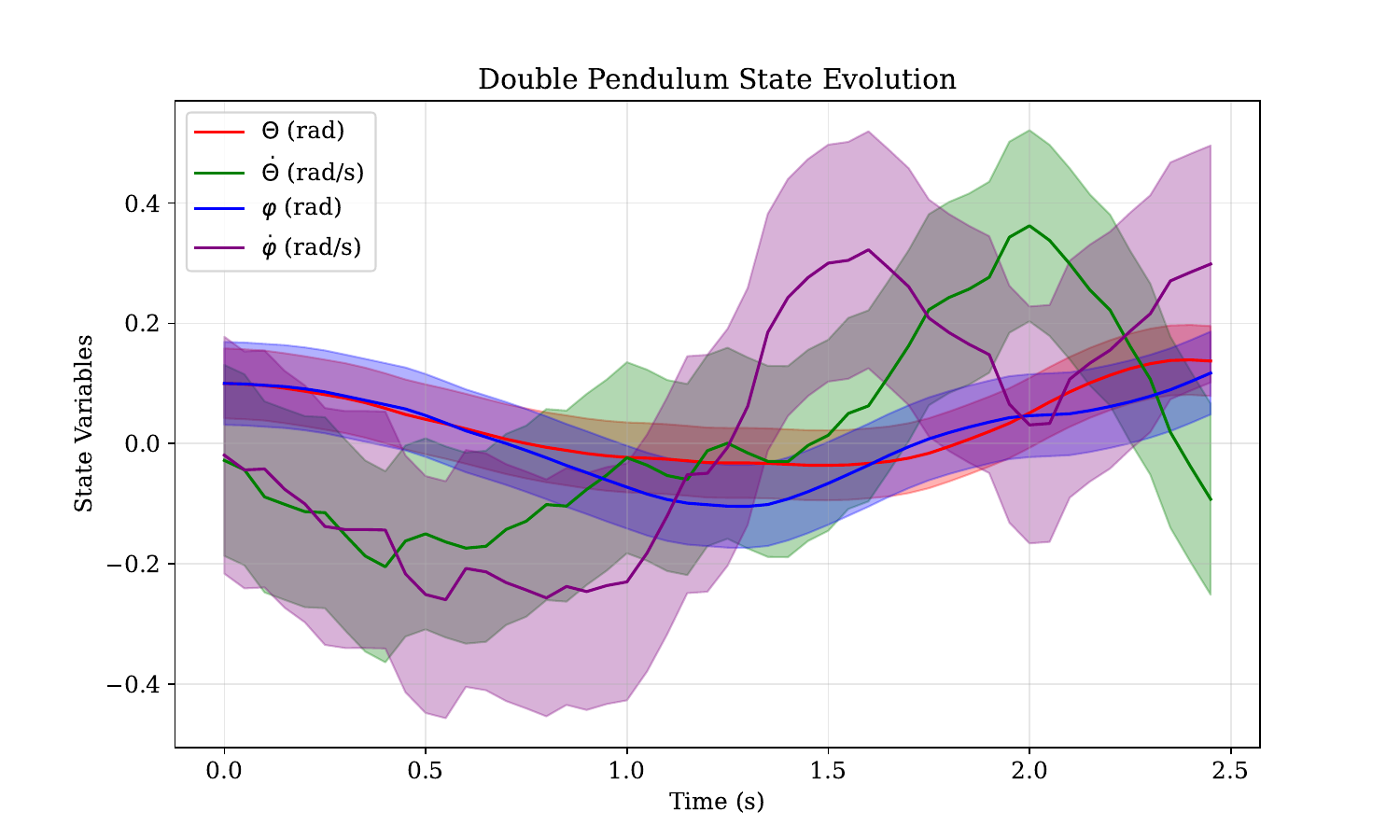}
\caption{Optimal trajectories for the state variables for the double pendulum.}
\label{fig21}
\end{figure}

\begin{figure}[H]
\centering
\includegraphics[width=1\linewidth]{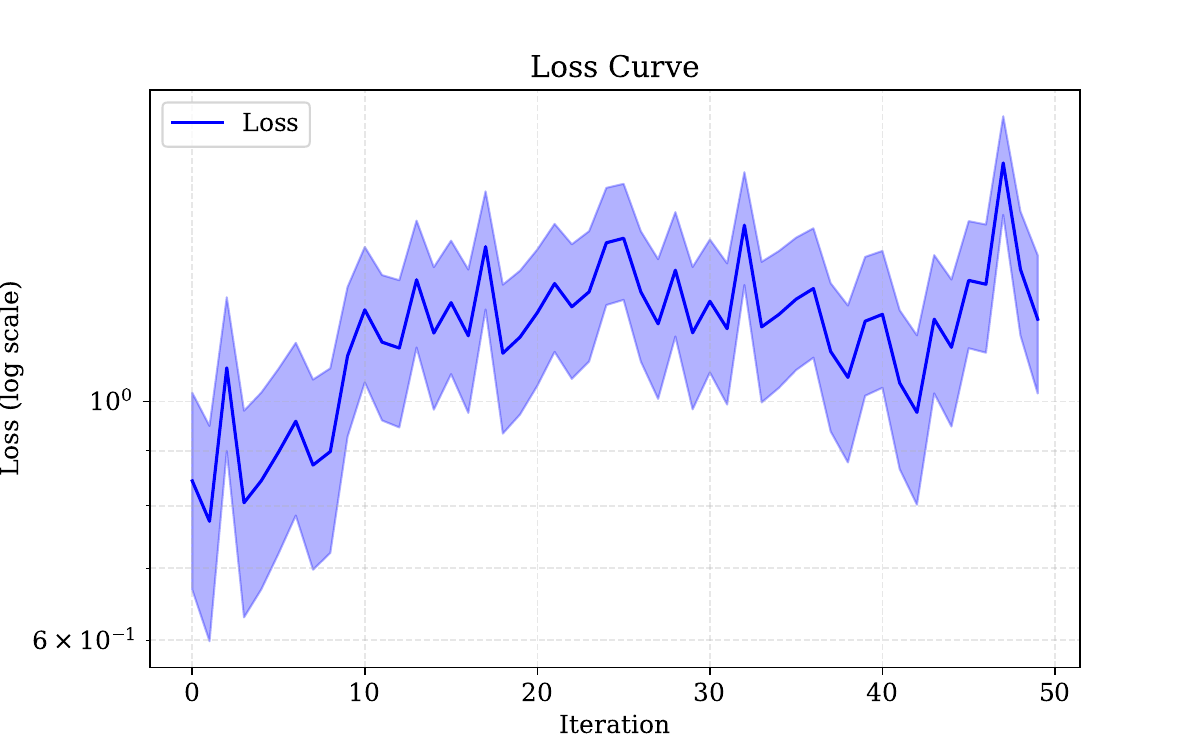}
\caption{Loss function for the double pendulum.}
\label{fig22}
\end{figure}

In Fig. \ref{fig19}, we observe that the quantum circuit consists of 4 qubits with the rotation gates providing the ansatz. As in previous experiments, the CNOT gates provide entanglement, and measurement is performed In the computational basis on the first two qubits. 

In Fig. \ref{fig20}, we see that both controls, $\tau_{\Theta}$ and $\tau_{\varphi}$, exhibit oscillatory behavior indicating that the controller is actively driving the system with periodic torque variations. The rapid resonance effects imply that the system is unstable and consumes excessive energy. In addition, the phase difference between the curves confirms that, indeed, the angles in the system are coupled. Lastly, there is a wide variance in both trajectories, as evidenced by the thick variance shading. 

In Fig. \ref{fig21}, we observe that for all four state variables, there are nonlinear trajectories, which is expected from a double pendulum system, which is known for chaotic dynamics. Further, we see that all four of these variables are highly correlated due to the phase differences, and it seems as though $\left(\dot{\Theta},\dot{\varphi}\right)$ leads $\left(\Theta,\dot{\varphi}\right)$ in phase. Further, the angular velocities $\left(\dot{\Theta},\dot{\varphi}\right)$ exhibit more rapid and irregular fluctuations, indicating that the controller is actively adjusting the system's motion to compensate for disturbances. 

In Fig. \ref{fig22}, we plot the loss function on a logarithmic scale, making it easier to observe trends and convergence, or lack thereof, behavior. The most striking feature of the plot is the lack of a clear downward, monotonically decreasing, trend. While there are fluctuations and some temporary dips, the overall loss does not consistently decrease. Moreover, it increases towards the end of the 50 iterations. This suggests that the optimization process is not converging to a minimum, and this QI-MPC approach struggles to find good control parameters. Thus, we conclude that for this compound pendulum system, this control strategy approach is suboptimal. We can attribute this behavior to several factors: The ineffectiveness of the QI-MPC approach for these types of problems and poor choice of the loss function \eqref{double pendulum loss}, among others. Due to the significant variance in the loss function throughout the training process, we hypothesize that the QI-MPC approach struggles with handling control tasks for oscillating systems. In Sec. \ref{physical results}, we investigate this hypothesis more deeply.   


\section{Physical Results}\label{physical results}
In this section, we present two important results that come out of observations from our experiments in Sec. \ref{experiments}. The first involves optimal control strategies, or the lack thereof, related to oscillatory systems, and the second is related to non-oscillatory systems; whether the dynamical model of the system is linear or nonlinear, given enough qubits, the hybrid QI-MPC would be the best model control strategy, and we identify applications where this result may and may not apply.  

\subsection{Nonfeasibility of hybrid QI-MPC Approaches to Control Oscillatory Systems}
We have observed in experiments 4 (Sec. \ref{experiment2}) and 5 (Sec. \ref{experiment3}), and to a certain degree in experiment 3 (Sec. \ref{autonomous vehicle experiment}) that the control variables exhibited highly erratic behavior, and in the former cases even the loss functions demonstrated such behavior. This is because these are oscillatory systems, and we hypothesize that a hybrid QI-MPC method is not the best control strategy for such systems. We formalize this in the Proposition \ref{prop1}. However, we have to demonstrate some preliminary results before establishing proof. 

\begin{theorem}\label{nyquist-shannon}
(Nyquist-Shannon Sampling Theorem). A band-limited signal, with the highest frequency component $f_{\max}$, can be perfectly reconstructed from its samples if it is sampled at a rate at least twice the highest frequency present in the signal. Mathematically, for a sampling rate frequency $f_{s}$,
\begin{equation*}
f_{s}\geq 2f_{\max}.
\end{equation*}
\end{theorem}

\begin{proof}
This is a standard result in Signal Processing, and proof can be found in any mathematically oriented textbook on the subject; see, for example, \cite{taub1991principles}.  
\end{proof}

\begin{lemma}\label{lemma1}
The feedback delay of hybrid QI-MPC methods exceeds the characteristic timescale of oscillatory dynamics.
\end{lemma}

\begin{proof}
Let $\omega_{\text{sys}}$ be the natural (angular) frequency of the oscillatory system. By theorem \ref{nyquist-shannon}, the updates of the control must be $f_{\text{control}}\geq 2f_{\text{sys}}=2\omega_{\text{sys}}/2\pi=\omega_{\text{sys}}/\pi$. Now, the various types of qubits operating frequencies are $\sim$ GHz (superconducting, spin qubits in semiconductors, Nitrogen-Vacancy centers), $\sim$ THz (trapped ions, photonic, optical transitions), or $\sim$ MHz (NMR, Rydberg transitions), but classical optimization schemes, that run on CPUs and GPUs, have clock speeds that operate in $\sim$ GHz. Thus, $f_{\text{classical optimizer}}\lesssim f_{\text{control}}$, violating theorem \ref{nyquist-shannon}.      
\end{proof}

\begin{lemma}\label{lemma2}
The energy scale of a quantum control is insufficient to stabilize macroscopic oscillators.
\end{lemma}

\begin{proof}
Mechanical energy, given as the sum of kinetic and potential energies, for everyday objects roughly scales as $E_{\text{mech}}\sim mgL\sim\mathcal{O}(10)\;\text{J}$ when the kinetic energy is small, for length everyday length scales $L$. Quantum control operates at Planck scale energies $E_{\text{qubit}}\sim \hbar f_{\text{qubit}}\sim\mathcal{O}(10^{-24})\;\text{J}$, with $\hbar\approx6.626\times10^{-34}\;\text{J.s}$. The ratio of these energies $E_{\text{qubit}}/E_{\text{mech}}\sim\mathcal{O}(10^{-25})$ is very small and therefore has a negligible effect on the control. Thus, the quantum control is negligible, as it provides little to no effect on stabilizing an oscillating system.  
\end{proof}

\begin{definition}\label{definition1}
(Contractible Sets). A set $S\in\left(Y,\mathcal{T}\right)$ (where $Y$ is a topological space endowed with topology $\mathcal{T}$) is contractable if there exists a continuous map $H:S\times\left[0,1\right]\to S$ such that:
\begin{enumerate}
\item \textbf{Identity map at time 0:} $H(x,0)=x\;\forall x\in S$.
\item \textbf{Map collapse to single point:} $H(x,1)=p\;\forall x\in S$ and some point $p\in S$. 
\end{enumerate}
\end{definition}

\begin{definition}\label{definition2}
(Reachable Set). For a system with dynamics $\dot{\mathbf{x}}_{k}=f(\mathbf{x}_{k},\mathbf{u}_{k})$, the reachable set $\mathcal{R}$, from an initial set $X_{0}$, is 
\begin{equation*}
\mathcal{R}=\left\{\mathbf{x}(t)|\mathbf{x}(0)\in X_{0},\mathbf{u}\in\mathcal{U}\right\},
\end{equation*}
with states $\mathbf{x}$, controls $\mathbf{u}$, and set of all admissible controls $\mathcal{U}$. 
\end{definition}

\begin{definition}\label{definition3}
(Non-Contractable Reachable Set). A set is $\mathcal{R}^{*}$ non-contractable reachable if:
\begin{enumerate}
\item It does not satisfy the definition of contractability articulated in definition \ref{definition1}.
\item If it is reachable, as articulated in definition \ref{definition2}. 
\end{enumerate}
\end{definition}

\begin{remark}\label{remark1}
Intuitively, a non-contractable reachable set is a set of states that a system can reach that has a ``hole'' or a topological structure that prevents it from being continuously shrunk down to a single point within itself.    
\end{remark}

\begin{definition}\label{definition4}
(Dense Covering). Let $\left(Y,\mathcal{T}\right)$ be a topological space and $\mathcal{C}$ be a collection of subsets in $Y$. $\mathcal{C}$ is a dense covering of $\left(Y,\mathcal{T}\right)$ if it satisfies:
\begin{enumerate}
\item Every point in $Y$ belongs to at least one set in the collection $\mathcal{C}$,
\begin{equation*}
\bigcup_{c \in \mathcal{C}}c=Y.
\end{equation*}
\item Each non-empty open set in $Y$ intersects with at least one of the sets in $\mathcal{C}$.
\begin{equation*}
\forall U \subseteq Y, U \neq \emptyset, U\;\text{is open},\exists c \in \mathcal{C}\;\text{such that}\;c\cap U \neq \emptyset.    
\end{equation*}
\end{enumerate}
\end{definition}

\begin{remark}\label{remark2}
A dense covering is a collection of sets that covers the entire space, and at least one of these sets ``reaches into'' every non-empty open region of the space.    
\end{remark}

\begin{definition}\label{definition5}
($l$-Torus). Let $S^{1}$ be the unit circle in $\mathbb{R}^{2}$,
\begin{equation*}
S^{1}=\left\{\left(x,y\right)|x^{2}+y^{2}=1\right\}.
\end{equation*}
A torus is the product of two unit circles,
\begin{equation*}
\mathbb{T}^{2}=\left\{\left(\left(x_{1},y_{1}\right),\left(x_{2},y_{2}\right)\right)|\left(x_{1},y_{1}\right)\in S^{1},\left(x_{2},y_{2}\right)\in S^{1}\right\}.
\end{equation*}
This definition can easily be extended to the $l$-dimensional torus
\begin{equation*}
\begin{aligned}
\mathbb{T}^{l}=&\;\underbrace{S^{1}\times S^{1}\times\ldots\times S^{1}}_{l-\text{times}}\\
=&\;\left\{\left(\left(x_{1},y_{1}\right),\ldots,\left(x_{l},y_{l}\right)\right)|\left(x_{1},y_{1}\right)\in S^{1},\ldots,\left(x_{l},y_{l}\right)\in S^{1}\right\}. 
\end{aligned}
\end{equation*}
\end{definition}

\begin{remark}\label{remark3}
In our proof of lemma \ref{lemma3} that follows, we will use a result from Kolmogorov–Arnold–Moser (KAM) theory without establishing proof, i.e. we take this result to be true and leave it up to the reader to verify. 

``For integrable Hamiltonian systems with $l$ degrees of freedom, their phase space is foliated by $l$-dimensional tori. The motion of these tori are quasiperiodic and can be expressed as a sum of periodic functions with frequencies that are, in general, incommensurable, having irrational ratios.'' A corollary to this statement, which is of interest to us in proving lemma \ref{lemma3}, is that if a system is oscillatory, it has tori in its phase space. 
\end{remark}

\begin{theorem}\label{theorem2}
(Brouwer's Invariance of Domain). Let $\emptyset\neq U\subseteq\mathbb{R}^{l}$, with $U$ being open, and let $f:U\to\mathbb{R}^{s}$ be a continuous injective map. Then:
\begin{enumerate}
\item $\l\geq s$.
\item $f(U)\subset\mathbb{R}^{s}$ and $f(U)$ is open. 
\end{enumerate}
The proof is beyond the scope of this research.    
\end{theorem}

\begin{remark}\label{Brouwer's theorem remark}
Theorem \ref{theorem2} implies that it is not possible to cover a higher-dimensional space with a lower-dimensional space in a meaningful way that preserves topological properties.
\end{remark}

\begin{lemma}\label{lemma3}
Quantum states, which are finite in dimension, cannot densely cover non-contractible reachable sets.  
\end{lemma}

\begin{proof}
Firstly, we note that from the result in KAM theory, articulated in remark \ref{remark3}, $l$-dimensional oscillatory systems have tori $\mathbb{T}^{l}$ in their phase space. Secondly, we know that in the hybrid approach, VQCs produce control $\mathbf{u}(t)=\bra{\psi(\boldsymbol{\theta})}\hat{\mathbf{u}}\ket{\psi(\boldsymbol{\theta})}\in\mathcal{H}$, the controls produced lie in a Hilbert space. Since a finite-dimensional Hilbert space of dimensional $n$ is homeomorphic (topologically equivalent) to the $2n$-dimensional real space, i.e. $\mathcal{H}^{n}\cong\mathbb{R}^{2n}$, and taking $l>2n$, by theorem \ref{theorem2}, we have that $\mathcal{H}$ cannot densely cover $\mathbb{T}^{l}$. 
\end{proof}

\begin{lemma}\label{lemma4}
Quantum noise disrupts control before mechanical damping acts.    
\end{lemma}

\begin{proof}
Firstly, we note that the timescales for mechanical damping (energy loss from the oscillating system that causes a gradual decrease in amplitude) are proportional to the inverse of the damping coefficient, i.e. $\uptau_{\text{mech}}\propto\gamma^{-1}\sim\mathcal{O}(10^{0})\;\text{s}\approx 1\;\text{s}$. Secondly, we note that the quantum coherence timescales $T_{2}$ are $\sim 10^{-1}$ s to $1$ year for nuclear spin qubits, $10^{-8}$ s to $10^{-6}$ s for electron spin qubits, $10^{9}$ s to $10^{-7}$ s for quantum dots; see \cite{BLACK2002189}. Thus, we conclude that $T_{2}\ll\uptau_{\text{mech}}$.      
\end{proof}

\begin{proposition}\label{prop1}
Hybrid quantum-classical (HQC) control methods cannot achieve globally stable Model Predictive Control (MPC) for nonlinear oscillatory systems with finite-dimensional quantum ans\"{a}tze, due to timescale, energy, topological, and decoherence constraints.    
\end{proposition}

\begin{proof}
We established each of the posits of proposition \ref{prop1} below.
\begin{enumerate}
\item The timescale constraint is established using lemma \ref{lemma1}. 
\item The energy scale constraint is established using lemma \ref{lemma2}.
\item The topological constraint is established using lemma \ref{lemma3}.
\item The decoherence constraint is established using lemma \ref{lemma4}.
\end{enumerate}
By 1--4 above, we see that the proposition holds true. 
\end{proof}

\begin{implication}
Proposition \ref{prop1} is an important result, and it establishes two important inferences:
\begin{enumerate}
\item For oscillatory systems, classical MPC methods, such as robust MCP, remain the best choice. 
\item Quantum control is viable for nanoscale systems where timescales and energies match.
\end{enumerate}
\end{implication}

\subsection{Outperformance of Classical MPC by QI-MPC}\label{outperformance of classical mpc}
We attempt to identify the types of systems that QI-MPC will and will not apply to. In practice, of course, Algorithm \ref{algo1} can be applied to any system, but how efficient the algorithm will be in determining the correct optimal control strategies is desired. Therefore, we establish a general result that can help with this. 

\begin{lemma}\label{lemma5}
A quantum computer can evaluate the MPC cost function over a superposition of control trajectories in parallel due to quantum parallelism. 
\end{lemma}

\begin{proof}
The control sequence is parametrized by the ansatz $U(\boldsymbol{\theta})$ for $\boldsymbol{\theta}\in\mathbb{R}^{m}$. Quantum amplitude estimation allows for the cost function associated with the ansatz to be evaluated $\mathcal{O}(1/\varepsilon)$ steps, whereas classical would take $\mathcal{O}(1/\varepsilon^{2})$ steps.  
\end{proof}

\begin{lemma}\label{lemma6}
An $n$-qubit system can represent $2^{n}$ classical states compactly.   
\end{lemma}

\begin{proof}
For a state space consisting of $N$ grid points, classical MPC requires $\mathcal{O}(N)$ memory to discretize the space. The hybrid QI-MPC approach can leverage the quantum mechanical property of superposition to discretize the state space in $\mathcal{O}(\log N)$ qubits.      
\end{proof}

\begin{lemma}\label{lemma7}
Hybrid MPC quantum gradient updates have a lower regret cost than classical MPC gradient updates.    
\end{lemma}

\begin{proof}
For a finite time horizon $T<\infty$ and a convex cost function, classical MPC methods have a regret cost of $\mathcal{O}(T)$. Quantum gradient update rules have a regret cost of $\mathcal{O}(\sqrt{T})$.    
\end{proof}

\begin{proposition}\label{prop2}
For a control system with a high-dimensional state space and sufficiently smooth dynamics, the QI-MPC scheme established in Algorithm \ref{algo1}, with a number of qubits scaling polynomially in the system dimension, can, in theory, achieve a lower regret bound than classical MPC methods, provided:  
\begin{enumerate}
\item The quantum circuit depth scales efficiently (polylogarithmically) in the horizon length.
\item The system dynamics are sufficiently smooth.
\item The control problem does not require ultra-fast feedback and the timescales are compatible with quantum measurement delays. 
\end{enumerate}
\end{proposition}

\begin{proof}
We observe that:
\begin{enumerate}
\item From lemma \ref{lemma5}, quantum parallelism speeds up the cost evaluation.   
\item Compact state representation, established in lemma \ref{lemma6}, lowers memory requirements.
\item Quantum optimization reduces the regret cost for control selection in classical MPC is reduced from linear to sublinearly when using a quantum controller to select control sequences, as established by lemma \ref{lemma7}. 
\end{enumerate}
Therefore, by the above three conditions, we conclude that proposition \ref{prop2} holds true. 
\end{proof}

We note that the type of systems that proposition \ref{prop2} holds for are systems analogous to experiments \ref{experiment1} and \ref{experiment4}. We generalize these to include:
\begin{enumerate}
\item \textbf{Multiagent Robotics:} Swarm control with many interacting agents that work both cooperatively, competitively, or a mixture of both. This is because quantum state compression allows efficient representation of high-dimensional policies, and smoothness ensures that the learning update rule converges efficiently.
\item \textbf{Systems with Low-Frequency Control Updates:} These include climate systems where the $\text{CO}_{2}$ needs to be controlled, portfolio optimization in Finance with daily updates, and supply chain optimization with low inventory dynamics. This is because it avoids quantum-classical loop delays, and there is no need for ultrafast feedback. 
\item \textbf{Quantum Mechanical or Nanoscale Systems:} These include systems where, for example, superconducting qubits need to be controlled as in Quantum Error Correction (QEC), atomic positioning and nanopositioning systems, and applications of molecular dynamics in situations like for example, protein folding. This is because these time and energy scales match the quantum hardware capabilities. 
\item \textbf{Convex Optimization and Optimization Problems with Well-Conditioned Landscapes:} These include classical Linear Quadratic Regulator (LQR) problems like aircraft flight control, spacecraft altitude control, rocket guidance, robotic arm motion, robot navigation, drone control, process control in plants, power system stabilization, and many others. In addition, problems that are strongly convex have a quadratic cost function. This is because quantum optimization performs well on smooth convex problems. 
\end{enumerate}

Contra to such problems enumerated above, we note that proposition \ref{prop2} does not hold for systems such as:
\begin{enumerate}
\item \textbf{Oscillatory Systems:} Such as pendula, vibrating systems, and even walking robots. As shown in proposition \ref{prop1}, these systems have a timescale mismatch (established by lemma \ref{lemma1} showing a violation of theorem \ref{nyquist-shannon}), amongst others.
\item \textbf{Systems that have Chaotic Dynamics:} These include weather systems, compound pendula, and fluid flow with high Reynolds number (turbulent flows). This is because quantum decoherence destroys control policies under exponential divergence.
\item \textbf{Systems with Ultrafast Dynamics:} These include stiff systems with widely separated timescales. Examples of such systems include control problems related to hypersonic flight and combustion engines. This is because quantum control loops cannot update at microsecond timescales.
\item \textbf{Systems with Hybrid Dynamics and Discontinuities:} These include gear shifting control and collision avoidance. This is because quantum optimization is known to struggle with non-smooth cost landscapes.
\end{enumerate}

\section{Conclusion}\label{conclusion}
We have presented our approach of selecting optimal control policies using a composite that uses QC for its computational efficiency and MPC for its robustness. The approach was substantiated using five numerical examples with differing degrees of complexity. A mixture of admissable and poor results were obtained for the experiments. In addition, we were able to provide safety guarantees and two important results in propositions \ref{prop1} and \ref{prop2} as a results of generalizing behaviors we saw when conducting the experiments. As a consequence of proposition \ref{prop2}, we were able to explicitly identify systems for which our proposed algorithm \ref{algo1} can potentially work and not work. We observed that the update rules for learning algorithms need to be more conservative, and optimizers need to be adaptive, for example, using Adam or RMSprop. 

For experiments \ref{experiment1}, \ref{experiment4}, and \ref{autonomous vehicle experiment} (with marginal effectiveness for experiment \ref{autonomous vehicle experiment}), we obtained sufficiently good results, and in experiments \ref{experiment2} and \ref{experiment3}, we obtained poor results. The results in the former three experiments can be attributed to the safety guarantees of Sec. \ref{safety guarantees}, specifically Formal Verification II, because these systems are closed-loop. The suboptimal outcomes of the latter two experiments are because they are oscillating systems and therefore invalidate the conditions of proposition \ref{prop2}.  

In proposition \ref{prop1}, we believe that our reasoning is sound under current technological and theoretical constraints for finite-dimensional HQC approaches. We note that alternative quantum representations like continuous-variable QC using quantum optics or bosonic modes, adaptive ans\"{a}tze and quantum-enhanced optimizers like QAOA might challenge its absolute validity. In a similar notion, we note that proposition \ref{prop2} is not universally valid, as we identified cases where validity holds and where validity does not hold, and only through experimentation one may determine where classical MPC methods are better than the hybrid approach, and vice versa. 

For more complicated dynamics where convergence is too slow, our approach can be improved by experimenting with the effect of adaptive gain with learning rates that have hyperbolic or power law decay so that the learning rate is reduced over time, facilitating finer control adjustments. 

In the quantum circuits used for selecting the optimal control policies, the CNOT connections formed a structured entanglement pattern across qubits, allowing information to be shared across the system. However, optimizing such circuits can be challenging due to barren plateaus that create regions in the parameter space where the gradients vanish. Thus, it would prove very interesting to investigate non-repeating entanglement patterns by judiciously placing the CNOT gates in the circuits. We note that the quantum circuits can be further refined by using alternative encoding schemes for mapping the classical data into quantum states, such as amplitude encoding, to potentially reduce the initial number of state rotations. Moreover, more efficient training schemes, such as layerwise techniques like quantum curriculum learning, could be used to improve training, and  hardware-efficient entanglement ans\"{a}tze could be used to improve trainability.



\end{document}